\documentclass[12pt,ruled,vlined,linesnumbered]{article}
\PassOptionsToPackage{natbib=true}{biblatex}
\usepackage[utf8]{inputenc}
\usepackage{geometry}
\usepackage{float}
\usepackage{algorithm2e}
\usepackage{amsmath}
\usepackage{amsthm}
\usepackage{amssymb}
\usepackage{graphicx}
\usepackage{setspace}
\usepackage[pdfusetitle,
 bookmarks=true,bookmarksnumbered=false,bookmarksopen=false,
 breaklinks=false,pdfborder={0 0 1},backref=false,colorlinks=false]
 {hyperref}

\makeatletter

\providecommand{\tabularnewline}{\\}

\newcommand{\lyxaddress}[1]{
	\par {\raggedright #1
	\vspace{1.4em}
	\noindent\par}
}
\theoremstyle{plain}
\newtheorem{assumption}{\protect\assumptionname}
\theoremstyle{plain}
\newtheorem{prop}{\protect\propositionname}
\theoremstyle{plain}
\newtheorem*{thm*}{\protect\theoremname}
\theoremstyle{plain}
\newtheorem*{cor*}{\protect\corollaryname}

\usepackage{pgfplots}
\usepackage{pgfplotstable}
\usepackage{xcolor}
\usepackage{tikz}
\usepackage{tkz-graph} % Make sure you have tkz-graph installed
\usetikzlibrary{patterns}
\usetikzlibrary{decorations.pathmorphing}
\usepackage{lscape}  % or \usepackage{pdflscape} for PDF output
\usepackage{array}  % for defining custom column types
\usepackage{hyperref}
\usepackage{booktabs} % For better table formatting
\usepackage{caption}  % For custom table captions
\usepackage{siunitx}  % For aligning numbers in the table
\usepackage{threeparttable}  % For adding notes to tables

\newcolumntype{C}[1]{>{\centering\arraybackslash}p{#1}}
\PassOptionsToPackage{position=top}{subfig}

\ifdefined\showcaptionsetup
 \PassOptionsToPackage{caption=false}{subfig}
\fi
\usepackage{subfig}
\makeatother

\usepackage[style=authoryear,style=authoryear, maxcitenames=3, maxbibnames=99, natbib=true, giveninits=true, uniquename=false]{biblatex}
\providecommand{\assumptionname}{Assumption}
\providecommand{\corollaryname}{Corollary}
\providecommand{\propositionname}{Proposition}
\providecommand{\theoremname}{Theorem}

\begin{document}
\title{Branching Fixed Effects: A Proposal for Communicating Uncertainty}
\author{Patrick Kline\thanks{This paper was prepared for the Econometric Society World Congress
2025. I thank Isaiah Andrews, Kevin Chen, Magne Mogstad, Rob Shimer,
Raffaele Saggio, Andres Santos, Chris Walters, Martin Weidner, and
Jinglin Yang for helpful comments and questions. This paper makes
use of the Veneto Work Histories dataset developed by the Economics
Department in Università Ca’ Foscari Venezia under the supervision
of Giuseppe Tattara.}\\
UC Berkeley}
\maketitle
\begin{abstract}
Economists often rely on estimates of linear fixed effects models
produced by other teams of researchers. Assessing the uncertainty
in these estimates can be challenging. I propose a form of sample
splitting for networks that partitions the data into statistically
independent \textit{branches}, each of which can be used to compute
an unbiased estimate of the parameters of interest in two-way fixed
effects models. These branches facilitate uncertainty quantification,
moment estimation, and shrinkage. Drawing on results from the graph
theory literature on tree packing, I develop algorithms to efficiently
extract branches from large networks. I illustrate these techniques
using a benchmark dataset from Veneto, Italy that has been widely
used to study firm wage effects.
\end{abstract}

\lyxaddress{\begin{center}
\par\end{center}}

\newpage{}

\section{Introduction}

Fixed effects methods have emerged as an important tool for scientific
communication in empirical economics, allowing researchers to summarize
complex empirical patterns found in vast administrative databases
with minimal loss of fidelity. A single research team can rarely envision,
much less exploit, all of the potential uses of such granular summaries.
It has therefore become routine for estimates to be publicly posted
online or hosted by statistical agencies for use by other researchers. 

For instance, the Opportunity Insights (OI) group provides public
estimates of intergenerational mobility at the census tract level
derived from fitting complex fixed effects models to US tax records
\citep{chetty2018impactsii}. Likewise, estimates of firm wage fixed
effects derived from German social security records have been developed
for secondary use by research teams in partnership with the Institute
for Employment Research \citep{card2015chk,bellmann2020akm}. Estimates
enabling secondary analysis of teacher and school value added have
also been shared between research teams \citep[e.g.,][]{chamberlain2013predictive}.

In each of these literatures, the workhorse models are linear and
involve two (or more) high-dimensional sets of fixed effects. Least
squares estimates of these parameters are unbiased but will inevitably
be noisy. Unfortunately, the properties of this noise are often poorly
documented. One problem is that conventional heteroscedasticity-robust
standard errors may be inconsistent when estimating models of moderately
high dimension, a problem that also afflicts the bootstrap \citep{cattaneo2018inference,el2018can}.
Another is that it is typically infeasible to release full covariance
matrices for more than a few thousand estimated effects, as the number
of distinct entries in such matrices grows quadratically. Moreover,
when the number of parameters being estimated is proportional to the
underlying sample size, the noise in the estimates may not be normally
distributed, in which case variance matrices provide incomplete guides
to uncertainty.

This ambiguity regarding the stochastic properties of published fixed
effects estimates presents difficulties for secondary analysis. For
example, \citet{mogstad2024inference} note that lack of information
on the covariance of the quasi-experimental place fixed effects developed
by \citet{chetty2018impactsii} is one factor that leads them to focus
their analysis on simpler cross-sectional estimates of intergenerational
mobility. In contrast to the OI place effects, estimates of firm wage
fixed effects are almost always provided without standard errors.
As discussed in \citet{kline2024firm}, researchers often treat these
estimated firm effects as outcomes in second-step regressions, reporting
(downstream) standard errors that may be severely biased by the implicit
assumption that the estimated effects are mutually independent.

This paper proposes breaking two-way fixed effects estimates into
statistically independent \textit{branches} as a means of transparently
communicating uncertainty. Each branch corresponds to a distinct subsample
of the microdata within which unbiased estimates of all target parameters
can be constructed. Building on a graph-theoretic interpretation of
two-way fixed effects models developed in \citet{kline2024firm},
the branches are constructed from edge-disjoint spanning trees of
the \textit{mobility network}---a graph representing the evolution
of group memberships in the data. Though I will focus on a setting
where this network captures the structure of worker moves between
firms, the spanning tree construction applies broadly to contexts
involving two-dimensional heterogeneity, including student--teacher,
patient--doctor, and judge--district pairings. Edges not included
in any spanning tree are appended to the last tree's branch, which
ensures the full-sample fixed effects estimator decomposes linearly
into branch-specific estimates.

Publicly releasing branches allows outside researchers to transparently
assess uncertainty in published fixed effect estimates and in projections
of those estimates onto observable covariates. Branches also simplify
more advanced tasks such as estimating moments of the fixed effects
and shrinkage that have proved challenging to extend to high-dimensional
fixed effects settings with heteroscedasticity \citep{kwon2023optimal,cheng2025optimal,robin2025}.
I illustrate these ideas using as an example the estimation of firm
wage fixed effects in a benchmark dataset from the Italian province
of Veneto. 

The idea of releasing branches of fixed effects has several precedents
in the literature. First, it mirrors the common practice among statistical
agencies of releasing replicates to assess uncertainty in published
estimates derived from surveys. For example, the US Bureau of Labor
Statistics uses two independent replicates of inflation measures in
each geographic area to evaluate the sampling variance of inflation
measures reported in the Consumer Price Index \citep{bls_cpi_calculation_2025}.
Likewise, the American Community Survey uses 80 replicates to assess
margins of error in published estimates \citep{census_acs_variance_replicate_tables_2025}.
While replicates are designed to be independent and identically distributed,
the branched estimates considered here will generally not be identically
distributed, as different subsamples will tend to exhibit different
noise levels. 

A second precedent comes from recent empirical work with two-way fixed
effects models exploiting random splits of the microdata. \citet{goldschmidt2017rise},
\citet{sorkin2018ranking}, \citet{drenik2023paying}, and \citet{jager2024worker}
all randomly split worker--firm datasets into half-samples in order
to obtain independent estimates of the subset of firm fixed effects
identified in both samples. Likewise, \citet{silver2021haste} uses
a split-sample approach to estimate the variance of physician value
added, while \citet{card2024industry} use a similar approach to estimate
the variance of industry wage effects. The popularity of random splitting
in these contexts derives from the simplicity of the covariance-based
methods that can be employed to account for estimation error. However,
an important limitation of random splitting approaches is that covariances
can only be computed among the random subset of parameters for which
multiple estimates turn out to be available. 

Rather than consider a random estimand, the branching approach deterministically
prunes and partitions the data so that the pertinent model parameters
remain estimable in each split. The partitioning strategy builds on
results from the graph theory literature on tree packing problems
\citep{nash1961edge,tutte1961decomposing,roskind_tarjan_1985}. First,
the microdata are restricted to the largest $k$-edge-connected component
($k$-ECC): a subgraph that remains connected after the removal of
any $k-1$ edges. In the worker--firm setting, restricting to a $k$-ECC
with $k>1$ is a natural refinement of the ubiquitous practice of
pruning to the largest connected component \citep{abowd2002computing}.
After pruning, the $k$-ECC is packed with edge-disjoint spanning
trees. These trees are then used to build branches that fully partition
the microdata.

Though the tree packing algorithm is deterministic, the typical graph
will admit a large number of alternative packings. To aid reproducibility,
it is useful to re-pack the graph many times. One can then compute
parameter estimates for alternate branch definitions based on different
packings of the graph. I discuss how to compute random packings and
find in an application that the standard deviation of parameter estimates
across branch definitions is often small enough to be addressed by
averaging over fewer than one hundred packings.

In addition to yielding a deterministic estimand, this \textit{prune-and-pack}
approach ensures that each fixed effect corresponding to a node in
the $k$-ECC has as many independent estimates as there are branches
of the graph. I show that the pooled two-way fixed effects estimator
can be written as a linear combination of branch-specific estimates.
These linear combinations, which amount to influence function contributions,
can be released alongside the branch-specific estimates to simplify
the quantification of uncertainty in the full-sample estimates.

The Veneto data provide a challenging test case for the branching
approach, as many firms are connected by a single worker move. In
this low connectivity environment, constructing even two branches
requires limiting attention to 29\% of the firms in the largest connected
set. However, the estimates turn out to be remarkably insensitive
to further pruning of the sample despite the fact that larger firms
tend to be better connected. Among other interesting findings, the
analysis reveals that the firm wage fixed effects in these pruned
samples are left-skewed and heavy-tailed. In other empirical literatures
where the relevant graphs tend to exhibit higher edge-connectivity---e.g.,
empirical work on teacher value added or models of place effects---it
may be feasible to extract dozens of branches without meaningfully
narrowing the scope of investigation.

The rest of the paper is structured as follows. The next section previews
the basic properties of branches and describes their potential uses
in greater detail. Section \ref{sec:Two-way-fixed-effects,} reviews
the algebra of two-way fixed effects estimators. Section \ref{sec:A-graph-representation}
introduces a graph-theoretic interpretation of the model. Section
\ref{subsec:Trees-and-branches} defines branches algebraically and
derives a representation of the fixed effects estimator as a linear
combination of branch-specific estimators. Section \ref{sec:Tree-packing}
reviews the theory underlying tree packing problems and outlines the
Prune-and-Pack algorithm. Section \ref{sec:Application} illustrates
the use of branches for quantifying uncertainty, moment estimation,
and shrinkage via exercises involving firm wage fixed effects derived
from administrative data in Veneto, Italy.

\section{The blessings of branches\protect\label{sec:Basic-idea}}

Before getting into the weeds of how to build branches, it is useful
to preview how branches can simplify many modern estimation tasks.
Suppose that we are interested in some $J\times1$ parameter vector
$\psi$ and have constructed a least squares estimator $\hat{\psi}\in\mathbb{R}^{J}$
of $\psi$ that is unbiased. 

Branches are formed by partitioning the microdata in a way that allows
the construction of $M\geq2$ mutually independent estimates $\left\{ \hat{\psi}_{b}\right\} _{b=1}^{M}$,
each of which obeys $\mathbb{E}\left[\hat{\psi}_{b}\right]=\psi$.
While each branch estimate has the same mean, their sampling distributions
may differ. In particular, the $J\times J$ variance matrix $\mathbb{V}\left[\hat{\psi}_{b}\right]$
is not assumed to be diagonal or identical across branches.

Since the branches correspond to disjoint subsamples, the full-sample
least squares estimator can be decomposed as the sum
\begin{equation}
\hat{\psi}=\sum_{b=1}^{M}\boldsymbol{C}_{b}\hat{\psi}_{b}\equiv\sum_{b=1}^{M}\hat{\phi}_{b},\label{eq:decomp}
\end{equation}
where each $\boldsymbol{C}_{b}$ is a known (i.e., fixed) $J\times J$
matrix and these matrices obey $\sum_{b=1}^{M}\boldsymbol{C}_{b}=\boldsymbol{I}$.
Note that each element of $\hat{\phi}_{b}\in\mathbb{R}^{J}$ is a
linear combination of the unbiased estimates in $\hat{\psi}_{b}\in\mathbb{R}^{J}$.
Thus, the full-sample fixed effects estimator effectively utilizes
each branch to estimate a different linear combination of $\psi$.
That is,
\[
\mathbb{E}\left[\hat{\phi}_{b}\right]\equiv\phi_{b}=\boldsymbol{C}_{b}\psi.
\]

It is also useful to construct leave-out estimators $\left\{ \hat{\psi}_{-b}\right\} _{b=1}^{M}$,
where each $\hat{\psi}_{-b}$ gives the least squares estimator fit
to the microdata in all branches but $b$. By assumption $\mathbb{E}\left[\hat{\psi}_{-b}\right]=\psi$.
Therefore, a statistically independent estimator of $\phi_{b}$ is
given by $\hat{\phi}_{-b}=\boldsymbol{C}_{b}\hat{\psi}_{-b}.$ As
detailed in Section \ref{subsec:Trees-and-branches}, after the data
have been partitioned into branches, computing each $\hat{\phi}_{b}$
and $\hat{\phi}_{-b}$ is no more difficult than computing $\hat{\psi}$. 

My proposal is for researchers to share the branch-level vectors $\left\{ \hat{\psi}_{b},\hat{\phi}_{b},\hat{\phi}_{-b}\right\} _{b=1}^{M}$
with the public. Note that $\hat{\psi}_{b}$, $\hat{\phi}_{b}$, and
$\hat{\phi}_{-b}$ are each $J\times1$ vectors that can be stored
as columns in a spreadsheet alongside the full-sample estimate $\hat{\psi}$.
By the decomposition property in \eqref{eq:decomp}, releasing the
full-sample estimate is redundant once columns corresponding to the
$\left\{ \hat{\phi}_{b}\right\} _{b=1}^{M}$ vectors have been provided.

Having outlined what branches are, I now preview how access to $\left\{ \hat{\psi}_{b},\hat{\phi}_{b},\hat{\phi}_{-b}\right\} _{b=1}^{M}$
can facilitate three common empirical tasks:\bigskip{}

1. \textbf{Quantifying uncertainty}

Since the branches are independent, the $J\times J$ variance matrix
of $\hat{\psi}$ can be written 
\[
\mathbb{V}\left[\hat{\psi}\right]=\sum_{b=1}^{M}\mathbb{V}\left[\hat{\phi}_{b}\right]=\sum_{b=1}^{M}\boldsymbol{C}_{b}\mathbb{V}\left[\hat{\psi}_{b}\right]\boldsymbol{C}_{b}'\equiv\Sigma.
\]

Branch independence and unbiasedness ensure that
\begin{eqnarray*}
\mathbb{E}\left[\hat{\phi}_{b}\left(\hat{\phi}_{b}-\hat{\phi}_{-b}\right)'\right] & = & \mathbb{E}\left[\hat{\phi}_{b}\hat{\phi}_{b}'\right]-\mathbb{E}\left[\hat{\phi}_{b}\hat{\phi}_{-b}'\right]\\
 & = & \left(\phi_{b}\phi_{b}'+\mathbb{V}\left[\hat{\phi}_{b}\right]\right)-\phi_{b}\phi_{b}'\\
 & = & \mathbb{V}\left[\hat{\phi}_{b}\right].
\end{eqnarray*}
Averaging $\hat{\phi}_{b}\left(\hat{\phi}_{b}-\hat{\phi}_{-b}\right)'$
with its transpose and summing across the $M$ branches yields a cross-fitting
variance estimator of the sort proposed by \citet{kline2020leave}:
\[
\hat{\Sigma}=\frac{1}{2}\sum_{b=1}^{M}\left[\left(\hat{\phi}_{b}-\hat{\phi}_{-b}\right)\hat{\phi}_{b}'+\hat{\phi}_{b}\left(\hat{\phi}_{b}-\hat{\phi}_{-b}\right)'\right].
\]
This estimator is symmetric $\left(\hat{\Sigma}=\hat{\Sigma}'\right)$
and unbiased $\left(\mathbb{E}\left[\hat{\Sigma}\right]=\Sigma\right)$
but need not be positive semi-definite in finite samples. 

The individual entries in $\hat{\Sigma}$ will tend to be imprecise
when $M$ is small. Fortunately, interest often centers on low-dimensional
quadratic functions of $\Sigma$, which are easier to estimate than
any particular entry in this matrix. Suppose, for example, that one
wishes to estimate the coefficient vector $\gamma$ from a projection
of $\psi$ onto the column space of a matrix of covariates $\boldsymbol{X}$.
Treating $\boldsymbol{X}$ as fixed and assuming that $\boldsymbol{S}_{xx}=\boldsymbol{X}'\boldsymbol{X}$
is full rank, we can write the estimand $\gamma=\boldsymbol{S}_{xx}^{-1}\boldsymbol{X}'\psi$
and the corresponding estimator $\hat{\gamma}=\boldsymbol{S}_{xx}^{-1}\boldsymbol{X}'\hat{\psi}$.
It follows that $\mathbb{E}\left[\hat{\gamma}\right]=\gamma$ and
$\mathbb{V}\left[\hat{\gamma}\right]=\boldsymbol{S}_{xx}^{-1}\boldsymbol{X}'\Sigma\boldsymbol{X}\boldsymbol{S}_{xx}^{-1}$.
Hence, an unbiased estimate of the variance of a second-step linear
projection is $\hat{\mathbb{V}}\left[\hat{\gamma}\right]=\boldsymbol{S}_{xx}^{-1}\boldsymbol{X}'\hat{\Sigma}\boldsymbol{X}\boldsymbol{S}_{xx}^{-1}$.
The entries of $\hat{\mathbb{V}}\left[\hat{\gamma}\right]$ will tend
to be significantly more precise than the entries of $\hat{\Sigma}$
when $J$ is large relative to $\dim\left(\boldsymbol{S}_{xx}\right)$.\footnote{Assume that $\lambda_{\min}\left(\boldsymbol{S}_{xx}\right)>\kappa>0$,
where $\lambda_{\min}\left(\boldsymbol{S}_{xx}\right)$ gives the
smallest eigenvalue of $\boldsymbol{S}_{xx}$. Then, as $J\rightarrow\infty$
(with $\dim\left(\boldsymbol{S}_{xx}\right)$ fixed), the noise in
$\hat{\mathbb{V}}\left[\hat{\gamma}\right]$ will become negligible
relative to the noise in $\hat{\Sigma}$.}

\bigskip{}

2. \textbf{Moment estimation}

Let $\odot$ denote the element-wise product operator. For any two
branches $b$ and $\ell$, independence implies $\mathbb{E}\left[\hat{\psi}_{b}\odot\hat{\psi}_{\ell}\right]=\psi\odot\psi$.
Hence, an unbiased estimate of the average squared entry in $\psi$
(i.e., the second uncentered moment of $\psi$) is $\frac{1}{J}\boldsymbol{1}'\left(\hat{\psi}_{b}\odot\hat{\psi}_{\ell}\right)$,
where $\boldsymbol{1}$ is a $J\times1$ vector of ones. A more precise
estimator can be had by averaging across all $\binom{M}{2}$ distinct
pairs of branches: 
\[
\frac{1}{J}\boldsymbol{1}'\left[\frac{2}{M\left(M-1\right)}\sum_{b=1}^{M}\sum_{\ell<b}\hat{\psi}_{b}\odot\hat{\psi}_{\ell}\right].
\]

Likewise, third moments can be estimated by averaging products over
all $\binom{M}{3}$ distinct triples of branches: 
\[
\frac{1}{J}\boldsymbol{1}'\left(\binom{M}{3}^{-1}\sum_{b_{1}=1}^{M}\sum_{b_{2}<b_{1}}\sum_{b_{3}<b_{2}}\hat{\psi}_{b_{1}}\odot\hat{\psi}_{b_{2}}\odot\hat{\psi}_{b_{3}}\right).
\]
By induction, moments up to order $M$ can be estimated via leveraging
moment conditions of the form:
\[
\mathbb{E}\left[\hat{\psi}_{1}\odot\hat{\psi}_{2}\odot\dots\odot\hat{\psi}_{M}\right]=\psi^{\circ M},
\]
where the $\circ M$ superscript denotes raising the entries of a
vector to the $M$th power elementwise. A general class of unbiased
estimators for weighted moments of order less than or equal to $M$
is provided in Section \ref{subsec:Moment-estimation}.

\bigskip{}

3. \textbf{Shrinkage}

Standard empirical Bayes shrinkage arguments are predicated on the
assumption that the distribution of noise is known ex-ante \citep{waltersEB2024}.
A recent paper by \citet{ignatiadis2023empirical} proposes a best
predictor based on independent and identically distributed replicates
that adapts to the unknown noise distribution. Insights from this
paper carry over to the problem of predicting the elements of $\psi$
given independent branch estimates $\left\{ \hat{\psi}_{b}\right\} _{b=1}^{M}$
that are not identically distributed. The proposed procedure, which
builds on a suggestion by \citet{krutchkoff1967supplementary}, consists
of running a series of regressions of each entry of $\hat{\psi}_{b}$
against the values of that entry in the other branches. This process
is repeated for each choice of $b$. Finally, the predicted values,
which generally shrink the noisy branch-specific estimates, are averaged. 

\section{The AKM model and its algebra\protect\label{sec:Two-way-fixed-effects,}}

Constructing branches involves carefully splitting the sample so as
to preserve estimability of the model. It is useful to review the
basic algebra of two-way fixed effects estimators using as our running
example the worker--firm setting of \citet{abowd1999high}, henceforth
AKM. The treatment here will depart minimally from the setup in \citet{kline2024firm}.
Suppose there are $N$ workers and $J$ firms. For simplicity, I will
assume there are only 2 time periods, that all workers switch employers
between these periods, and that there are no time varying covariates.\footnote{Adjustments for covariates can be handled in a first step, in which
case the relevant outcomes $y_{it}$ becomes residualized wages. The
AKM model has no implications for wage dynamics within a worker--firm
match. Therefore, when additional time periods are available, nothing
is lost by collapsing $y_{it}$ down to worker--firm match means.} The AKM model can be written:
\[
y_{it}=\alpha_{i}+\psi_{\mathbf{j}\left(i,t\right)}+\varepsilon_{it}
\]
where $y_{it}$ is the log wage of worker $i\in\left\{ 1,2,\dots,N\right\} \equiv\left[N\right]$
in period $t\in\left\{ 1,2\right\} $ and the function $\mathbf{j}:\left[N\right]\times\left\{ 1,2\right\} \rightarrow\left\{ 1,2,\dots,J\right\} \equiv\left[J\right]$
gives the identity of the firm that worker $i$ is paired with in
period $t$. The $\left\{ \alpha_{i}\right\} _{i=1}^{N}$ are person
effects that can be ported from one employer to another, whereas the
$\left\{ \psi_{j}\right\} _{j=1}^{J}$ capture firm effects that must
be forfeited when leaving employer $j$. Both the person and firm
effects are parameters -- i.e., they are fixed effects. In contrast,
the time-varying errors $\left\{ \varepsilon_{it}\right\} _{i=1,t=1}^{N,2}$
are stochastic. Following the convention in the literature, the errors
are assumed to each have mean zero, which amounts to a strict exogeneity
requirement.

Interest often centers on the firm effects, which, under assumptions
described in \citet{kline2024firm}, can be thought of as capturing
average treatment effects on wages of moving between particular firm
pairs. Adopting this perspective, we can eliminate the person effects
with a first differencing transformation
\begin{equation}
y_{i2}-y_{i1}=\psi_{\mathbf{j}\left(i,2\right)}-\psi_{\mathbf{j}\left(i,1\right)}+\underbrace{\varepsilon_{i2}-\varepsilon_{i1}}_{\equiv u_{i}}.\label{eq:FD}
\end{equation}
Thus, each worker's wage change offers a noisy estimate of the difference
in firm effects between an origin firm $\mathbf{j}\left(i,1\right)$
and a destination firm $\mathbf{j}\left(i,2\right)$. The noise term
$u_{i}=\varepsilon_{i2}-\varepsilon_{i1}$ captures idiosyncratic
differences across workers in the wage changes they experience when
making this transition. These differences might reflect treatment
effect heterogeneity across workers in the effects of switching firms
or idiosyncratic shocks that would have occurred even in the absence
of the move (e.g., a health shock). Accordingly, the notion of uncertainty
that we seek to address is how estimates of the firm effects might
change if a different set of noise contributions $\left\{ u_{i}\right\} _{i=1}^{N}$
had been drawn. 

Denote the set of origin-destination pairs traversed by workers between
the two periods as 
\[
\mathcal{P}=\left\{ \left(o,d\right)\in\left[J\right]^{2}:\mathbf{j}\left(i,1\right)=o,\mathbf{j}\left(i,2\right)=d\text{ \ensuremath{\text{for some \ensuremath{i\in\left[N\right]}}}}\text{ and }o\neq d\right\} .
\]
The number of ordered pairs in this set is $\left|\mathcal{P}\right|$.
The \textit{unordered} pair of firms associated with worker $i$ will
be denoted $\left\{ \mathbf{j}\left(i,1\right),\mathbf{j}\left(i,2\right)\right\} $.
As detailed in Section \ref{subsec:Trees-and-branches}, such unordered
pairs will serve as the building blocks of branches. In particular,
any two workers moving between the same firms (in either direction)
will be assigned to the same branch. 

To ensure independence across branches, it suffices to assume independence
of the noise across distinct unordered firm pairs.
\begin{assumption}[Dyad independence\label{Pair-independence}]
For all $\ensuremath{\left(o,d\right)\in}\mathcal{P}$,
\[
\left\{ u_{i}:\left\{ \mathbf{j}\left(i,1\right),\mathbf{j}\left(i,2\right)\right\} =\left\{ o,d\right\} \right\} \perp\left\{ u_{i}:\left\{ \mathbf{j}\left(i,1\right),\mathbf{j}\left(i,2\right)\right\} \neq\left\{ o,d\right\} \right\} .
\]
\end{assumption}
By construction, each unordered pair $\left\{ o,d\right\} $ is assigned
to exactly one branch. Hence, Assumption \ref{Pair-independence}
implies mutual independence across branches.

Note that the dyad independence assumption allows workers moving between
the same firms to exhibit correlated noise. Such correlation could
arise, for example, if worker departures from a firm $o$ to a firm
$d$ are prompted by a shared human capital shock (e.g., a team of
employees is poached by a better paying firm in response to their
accomplishments at the origin firm). In practice, Section \ref{sec:Application}
will demonstrate that branch-based estimates of variance components
are remarkably similar to published estimates derived from the stronger
assumption that the $\left\{ u_{i}\right\} _{i=1}^{N}$ are mutually
independent across workers. A sampling based justification for the
traditional mutually independent noise representation is pursued in
\ref{sec:Independence-and-uncertainty}. 

Dyad independence could be violated by a shock that strikes all employees
of a firm in a single period. For example, if period 1 is unusually
productive for firm $o$ relative to period 2, we might expect the
wages offered by that firm to be higher in period 1. In such a case,
firm $o$'s wage effect itself would become unstable. Recent work
by \citet{lachowska2023firm} and \citet{engbom2023firm} provides
evidence that firm wage effects are extremely persistent, suggesting
transitory firm shocks are unlikely to generate quantitatively meaningful
violations of Assumption \ref{Pair-independence} over short time
horizons.

\subsection{The least squares estimator}

It is useful to write \eqref{eq:FD} in matrix notation as
\begin{eqnarray}
\boldsymbol{y}_{2}-\boldsymbol{y}_{1} & = & \left(\boldsymbol{F}_{2}-\boldsymbol{F}_{1}\right)\psi+\boldsymbol{u}\nonumber \\
 & \equiv & \boldsymbol{P}\boldsymbol{B}'\psi+\boldsymbol{u}\mathrm{,}\label{eq:dy}
\end{eqnarray}
where $\boldsymbol{F}_{t}$ is an $N\times J$ matrix of firm indicators
with $i$th row and $j$th column entry $1\left\{ \mathbf{j}\left(i,t\right)=j\right\} $,
$\psi$ is a $J\times1$ vector of firm effects, and $\boldsymbol{u}=\left(u_{1},\dots,u_{N}\right)^{\prime}$
is an $N\times1$ vector of noise contributions.

The $J\times\left|\mathcal{P}\right|$ matrix $\boldsymbol{B}$ is
known in the graph theory literature as the (signed) incidence matrix.
Each row of $\boldsymbol{B}$ corresponds to a particular firm, while
each column corresponds to an origin-destination pair. Denoting the
$p$th ordered pair in the set $\mathcal{P}$ by $\left(o_{p},d_{p}\right)$,
the entry in the $j$th row and $p$th column of $\boldsymbol{B}$
can be written $1\left\{ d_{p}=j\right\} -1\left\{ o_{p}=j\right\} $.
Note that each column of $\boldsymbol{B}$ has exactly two non-zero
entries, one of which takes the value $-1$, representing departure
from some origin firm, and one of which takes the value $+1$, representing
arrival at a destination firm. 

The $N\times\left|\mathcal{P}\right|$ matrix $\boldsymbol{P}$ consists
of origin-destination pair indicators. The entry in the $i$th row
and $p$th column of $\boldsymbol{P}$ can be written $1\left\{ \mathbf{j}\left(i,1\right)=o_{p},\mathbf{j}\left(i,2\right)=d_{p}\right\} $.
Thus, 
\[
\hat{\Delta}=\left(\boldsymbol{P}'\boldsymbol{P}\right)^{-1}\boldsymbol{P}'\left(\boldsymbol{y}_{2}-\boldsymbol{y}_{1}\right)
\]
gives the $\left|\mathcal{P}\right|\times1$ vector of mean wage changes
between each origin-destination pair. Note that $\boldsymbol{P}'\boldsymbol{P}\equiv\boldsymbol{W}$
is a diagonal $\left|\mathcal{P}\right|\times\left|\mathcal{P}\right|$
matrix giving the number of workers who move from each origin to each
destination firm. Hence, $\boldsymbol{P}'\left(\boldsymbol{y}_{2}-\boldsymbol{y}_{1}\right)=\boldsymbol{W}\hat{\Delta}$.
In what follows, the matrices $\boldsymbol{B}$ and $\boldsymbol{P}$
will be treated as fixed.

The identity $\boldsymbol{F}_{2}-\boldsymbol{F}_{1}=\boldsymbol{P}\boldsymbol{B}'$
provides a useful way to separate the sorts of moves present in the
data from how many workers move between each firm pair. Premultiplying
the system in \eqref{eq:dy} by $\boldsymbol{B}\boldsymbol{P}'$ yields
\[
\boldsymbol{B}\boldsymbol{W}\hat{\Delta}=\boldsymbol{B}\boldsymbol{W}\boldsymbol{B}'\psi+\boldsymbol{B}\boldsymbol{P}'\boldsymbol{u}.
\]
 Strict exogeneity of the errors implies $\mathbb{E}\left[\boldsymbol{u}\right]=0$.
This yields the moment condition
\[
\mathbb{E}\left[\boldsymbol{B}\boldsymbol{W}\hat{\Delta}\right]=\boldsymbol{B}\boldsymbol{W}\boldsymbol{B}'\psi\equiv\boldsymbol{L}\psi.
\]

The $J\times J$ matrix $\boldsymbol{L}=\boldsymbol{B}\boldsymbol{W}\boldsymbol{B}'$
is known in graph theory as the (weighted) Laplacian matrix. When
the mobility network is connected -- a concept we will review in
more detail below -- the Laplacian will have rank $J-1$, implying
$\psi$ is identified up to a constant. It is common to resolve this
indeterminacy by normalizing one of the firm effects to zero. I will
follow this convention by setting the first entry of $\psi$ to zero,
which yields the reduced system
\[
\mathbb{E}\left[\boldsymbol{B}_{\left(1\right)}\boldsymbol{W}\hat{\Delta}\right]=\boldsymbol{\boldsymbol{L}}_{\left(1\right)}\psi_{\left(1\right)},
\]
where $\boldsymbol{B}_{\left(1\right)}$ is the submatrix obtained
by removing the first row from $\boldsymbol{B}$, $\psi_{\left(1\right)}$
is the vector of length $J-1$ formed by removing the first entry
from $\psi$, and $\boldsymbol{L}_{\left(1\right)}=\boldsymbol{B}_{\left(1\right)}\boldsymbol{W}\boldsymbol{B}'_{\left(1\right)}$
is the matrix formed by dropping the first row and column of $\boldsymbol{L}$.
Solving this reduced system for $\psi_{\left(1\right)}$ yields the
least squares estimator
\begin{equation}
\hat{\psi}_{\left(1\right)}=\boldsymbol{\boldsymbol{L}}_{\left(1\right)}^{-1}\boldsymbol{B}_{\left(1\right)}\boldsymbol{W}\hat{\Delta}.\label{eq:OLS}
\end{equation}
Thus, the estimated firm effects are a Laplacian normalized linear
combination of the average wage changes of movers along all origin-destination
pairs.

\subsection{Pooling data on firm pairs}

When worker moves are present in both directions between a pair of
firms, some columns of the incidence matrix $\boldsymbol{B}$ will
be mirror images of each other. In such cases, it is possible to further
simplify \eqref{eq:OLS} by expressing the estimated firm effects
as a linear combination of a shorter vector of oriented average wage
changes $\overrightarrow{\Delta}$ that difference the entries of
$\hat{\Delta}$ along opposite directions of worker flow. In addition
to simplifying computation of $\hat{\psi}$, this pooled representation
will provide a foundation for the next section, which develops an
interpretation of the AKM model as an undirected graph. As detailed
in Section \ref{subsec:Trees-and-branches}, each branch-specific
estimate $\hat{\psi}_{b}$ will ultimately correspond to a linear
combination of a mutually exclusive subset of the entries of $\overrightarrow{\Delta}$.

To illustrate the basic idea, suppose that our data only measure moves
between two firms. In such a case, we can write $\hat{\Delta}=\left(\hat{\Delta}_{+},\hat{\Delta}_{-}\right)^{\prime}$,
where $\hat{\Delta}_{+}$ is the average wage change of workers moving
from firm 2 to firm 1, and $\hat{\Delta}_{-}$ is the average wage
change of workers moving from firm 1 to firm 2. Since $\mathbb{E}\left[\hat{\Delta}_{+}\right]=\psi_{1}-\psi_{2}$
and $\mathbb{E}\left[\hat{\Delta}_{-}\right]=\psi_{2}-\psi_{1}$,
it is natural to pool this information. Letting $n_{+}$ be the number
of movers from firm 2 to firm 1 and $n_{-}$ the number of movers
from firm $1$ to firm $2$, we can define $\overrightarrow{\Delta}=\left(n_{+}\hat{\Delta}_{+}-n_{-}\hat{\Delta}_{-}\right)/\left(n_{+}+n_{-}\right)$
(scalar in this toy example), which provides an unbiased estimate
of $\psi_{1}-\psi_{2}$. This mover-weighted average provides an efficient
pooling of the information in $\hat{\Delta}_{+}$ and $\hat{\Delta}_{-}$
under the auxiliary assumption that the worker level errors $\boldsymbol{u}$
are homoscedastic. 

When $Q\leq\left|\mathcal{P}\right|/2$ firm pairs have worker flows
in both directions, the incidence matrix can be partitioned as 
\[
\underbrace{\boldsymbol{B}}_{J\times\left|\mathcal{P}\right|}=\left[\underbrace{\boldsymbol{D}}_{J\times Q},\underbrace{-\boldsymbol{D}}_{J\times Q},\underbrace{\boldsymbol{R}}_{J\times\left(\left|\mathcal{P}\right|-2Q\right)}\right],
\]
where the matrices $\left(\boldsymbol{D},-\boldsymbol{D}\right)$
capture moves in opposite directions between firm pairs and the residual
incidence matrix $\boldsymbol{R}$ represents moves between firm pairs
that only experience flows in one direction. Note that this representation
is not unique: there are $2^{Q}$ possible choices of $\boldsymbol{D}$,
each of which amounts to treating one member of the firm pair as the
origin and the other as the destination. 

Given any orientation $\boldsymbol{D}$, the vector of wage changes
between origin-destination pairs can be partitioned
\[
\underbrace{\hat{\Delta}}_{\left|\mathcal{P}\right|\times1}=\left[\underbrace{\hat{\Delta}_{+}^{\prime}}_{1\times Q},\underbrace{\hat{\Delta}_{-}^{\prime}}_{1\times Q},\underbrace{\hat{\Delta}_{R}^{\prime}}_{1\times\left(\left|\mathcal{P}\right|-2Q\right)}\right]^{\prime},
\]
where $\hat{\Delta}_{+}$ gives the average wage changes of workers
moving from origins to destinations dictated by $\boldsymbol{D}$,
$\hat{\Delta}_{-}$ gives the average wage changes workers moving
between those same firms in the opposite direction, and $\hat{\Delta}_{R}$
gives the average wage changes associated with $\boldsymbol{R}$.
The corresponding weight matrix can be partitioned 
\[
\underbrace{\boldsymbol{W}}_{\left|\mathcal{P}\right|\times\left|\mathcal{P}\right|}=\left[\begin{array}{ccc}
\underbrace{\boldsymbol{W}_{+}}_{Q\times Q} & 0 & 0\\
0 & \underbrace{\boldsymbol{W}_{-}}_{Q\times Q} & 0\\
0 & 0 & \underbrace{\boldsymbol{W}_{R}}_{\left(\left|\mathcal{P}\right|-2Q\right)\times\left(\left|\mathcal{P}\right|-2Q\right)}
\end{array}\right],
\]
where $\boldsymbol{W}_{+}$ captures the number of workers moving
in one direction between a pair and $\boldsymbol{W}_{-}$ the number
moving in the opposite direction. Hence, $\boldsymbol{B}_{\left(1\right)}=\left[\boldsymbol{D}_{\left(1\right)},-\boldsymbol{D}_{\left(1\right)},\boldsymbol{R}_{\left(1\right)}\right]$
and $\boldsymbol{B}_{\left(1\right)}\boldsymbol{W}\hat{\Delta}=\boldsymbol{D}_{\left(1\right)}\left(\boldsymbol{W}_{+}\hat{\Delta}_{+}-\boldsymbol{W}_{-}\hat{\Delta}_{-}\right)+\boldsymbol{R}_{\left(1\right)}\boldsymbol{W}_{R}\hat{\Delta}_{R}$. 

Plugging the latter expression into \eqref{eq:OLS} reveals that the
OLS estimator can be written
\begin{eqnarray}
\hat{\psi}_{\left(1\right)} & = & \boldsymbol{L}_{\left(1\right)}^{-1}\left[\boldsymbol{D}_{\left(1\right)}\left(\boldsymbol{W}_{+}\hat{\Delta}_{+}-\boldsymbol{W}_{-}\hat{\Delta}_{-}\right)+\boldsymbol{R}_{\left(1\right)}\boldsymbol{W}_{R}\hat{\Delta}_{R}\right]\nonumber \\
 & = & \boldsymbol{L}_{\left(1\right)}^{-1}\left[\begin{array}{cc}
\boldsymbol{D}_{\left(1\right)} & \boldsymbol{R}_{\left(1\right)}\end{array}\right]\left[\begin{array}{cc}
\boldsymbol{W}_{D} & 0\\
0 & \boldsymbol{W}_{R}
\end{array}\right]\left(\begin{array}{c}
\hat{\Delta}_{D}\\
\hat{\Delta}_{R}
\end{array}\right)\nonumber \\
 & \equiv & \boldsymbol{L}_{\left(1\right)}^{-1}\underbrace{\bar{\boldsymbol{B}}_{\left(1\right)}}_{\left(J-1\right)\times\left(\left|\mathcal{P}\right|-Q\right)}\underbrace{\boldsymbol{\bar{W}}}_{\left(\left|\mathcal{P}\right|-Q\right)\times\left(\left|\mathcal{P}\right|-Q\right)}\underbrace{\overrightarrow{\Delta}}_{\left(\left|\mathcal{P}\right|-Q\right)\times1},\label{eq:undirected}
\end{eqnarray}
where $\boldsymbol{W}_{D}=\boldsymbol{W}_{-}+\boldsymbol{W}_{+}$
records the total number of workers moving between each of the $Q$
firm pairs with flows in both directions and $\hat{\Delta}_{D}=\boldsymbol{W}_{D}^{-1}\left(\boldsymbol{W}_{+}\hat{\Delta}_{+}-\boldsymbol{W}_{-}\hat{\Delta}_{-}\right)$
gives the average wage change between these firm pairs across an orientation
$\boldsymbol{D}$. 

Equation \eqref{eq:undirected} reveals that no information is lost
by collapsing the vector $\hat{\Delta}$ of average wage changes between
origin-destination pairs down to a shorter vector $\overrightarrow{\Delta}$
of oriented average wage changes between firm pairs. Importantly,
this representation holds for any choice of $\boldsymbol{D}$. The
Laplacian is likewise invariant to the choice of orientation $\boldsymbol{D}$
as:
\[
\boldsymbol{L}_{\left(1\right)}=\boldsymbol{B}_{\left(1\right)}\boldsymbol{W}\boldsymbol{B}_{\left(1\right)}^{\prime}=\boldsymbol{D}_{\left(1\right)}\left(\boldsymbol{W}_{+}+\boldsymbol{W}_{-}\right)\boldsymbol{D}_{\left(1\right)}^{\prime}+\boldsymbol{R}_{\left(1\right)}\boldsymbol{W}_{R}\boldsymbol{R}_{\left(1\right)}^{\prime}=\bar{\boldsymbol{B}}_{\left(1\right)}\boldsymbol{\bar{W}}\bar{\boldsymbol{B}}_{\left(1\right)}^{\prime}.
\]
The next section details how the Laplacian can be used to construct
an undirected graphical representation of worker mobility patterns.

\section{A graph-theoretic interpretation\protect\label{sec:A-graph-representation}}

I will now introduce a graph-theoretic interpretation of the AKM model
and the Laplacian $\boldsymbol{L}$, which we have already seen plays
a key role in determining estimability of the firm effects. By clarifying
when the least squares estimator $\hat{\psi}_{\left(1\right)}$ can
be computed, this discussion will provide a principled approach to
partitioning the sample into branches. 

Define the mobility graph $G=\left(V,E\right)$ of a two-way fixed
effects model as a finite set of vertices $V$ and a collection of
edges $E\subseteq\{(s,v)\in V\times V:s\neq v\}$ that connect pairs
of vertices. I will refer to individual edges by $e_{sv}$. In contrast
to the treatment in \citet{kline2024firm}, the edges will be viewed
here as undirected, which implies that $e_{sv}=e_{vs}$.

\begin{figure}[H]
\centering
\centering{}% Requires tikz and positioning
\usetikzlibrary{matrix, positioning}

\begin{tikzpicture}[scale=1.2,
    node_style/.style={circle, draw=black, thick, minimum size=10mm, font=\sffamily},
    ecc_3_node/.style={node_style, fill=red!20},
    ecc_1_node/.style={node_style},
    tree1_edge/.style={blue, very thick, dash pattern=on 6pt off 2pt},
    tree2_edge/.style={red, ultra thick, dash pattern=on 1pt off 8pt},
    ecc_edge/.style={thick, black!70},
    edge_id/.style={font=\small, right=0.3}
]

% Node 1 outside the 2-ECC (moved down-left)
\node[ecc_1_node] (N1) at (0,0)        {1};

% 3-ECC nodes with exaggerated rotation arrangement
\node[ecc_3_node] (N2) at (2,1)       {2};
\node[ecc_3_node] (N3) at (1.0,4.0)    {3};
\node[ecc_3_node] (N4) at (4.0,5.0)    {4};
\node[ecc_3_node] (N5) at (6.0,2.5)    {5};

% 3-ECC edges (IDs only)
\draw[ecc_edge] (N2) -- (N3) node[midway, edge_id] {$e_{23}$};
\draw[ecc_edge] (N2) -- (N4) node[midway, edge_id, , yshift=-8pt] {$e_{24}$};
\draw[ecc_edge] (N2) -- (N5) node[midway, edge_id] {$e_{25}$};
\draw[ecc_edge] (N3) -- (N5) node[midway, edge_id] {$e_{35}$};
\draw[ecc_edge] (N4) -- (N5) node[midway, edge_id] {$e_{45}$};
\draw[ecc_edge] (N3) -- (N4) node[midway, edge_id] {$e_{34}$};

% Edge connecting node 1
\draw[ecc_edge] (N1) -- node[midway, edge_id] {$e_{12}$} (N2);

% Spanning tree 1 (blue dashed)
\draw[tree1_edge] (N2) -- (N3)   (N3) -- (N5)   (N5) -- (N4);

% Spanning tree 2 (red dotted)
\draw[tree2_edge] (N2) -- (N4)   (N2) -- (N5);

% Overlay dotted red on edge e34
\draw[tree2_edge] (N3) -- (N4);

% Legend east of the figure
\matrix (legend) [
    matrix of nodes,
    nodes={anchor=mid west,font=\small},
    column sep=1em,
    row sep=0.4em,
    anchor=west
] at ([xshift=1cm]current bounding box.east) {
  \tikz \draw[tree1_edge] (0,0) -- (6mm,0); & Tree 1 (dashed) \\
  \tikz \draw[tree2_edge] (0,0) -- (6mm,0); & Tree 2 (dotted) \\
  \node[ecc_3_node, minimum size=6mm] {}; & Vertices in 3-ECC \\
};

\end{tikzpicture}
\caption{A connected mobility graph with $J=5$ and $\left|E\right|=7$\protect\label{fig:mobility=000020graph}}
\end{figure}
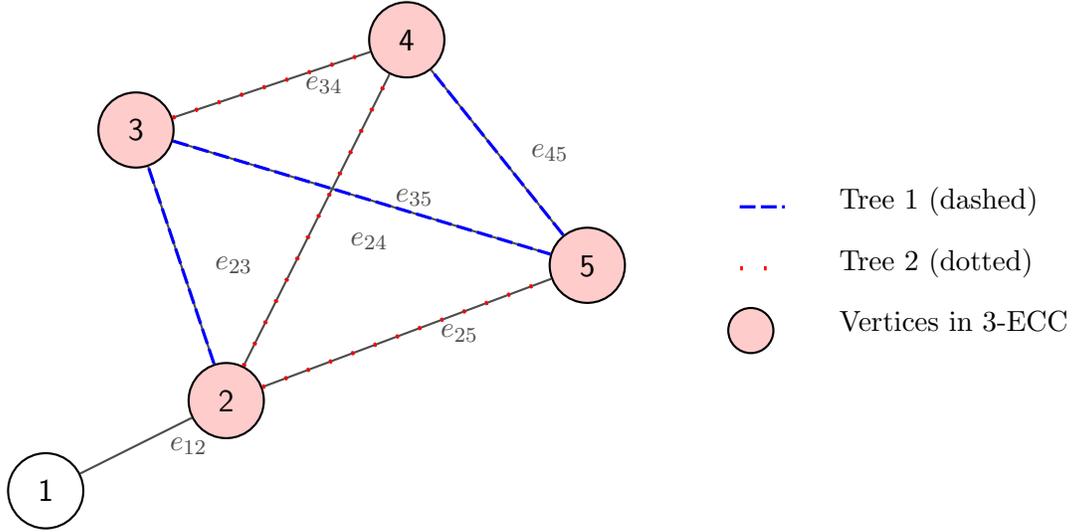

In the AKM model, the vertices are firms. Hence, $V=\left[J\right]$
and $\left|V\right|=J$. An edge joins a pair of firms $s$ and $v$
whenever at least one worker moves between the pair in either direction.
That is, in our earlier panel notation, 
\[
E=\{\left(s,v\right)\in\left[J\right]^{2}:\left\{ \mathbf{j}\left(i,1\right),\mathbf{j}\left(i,2\right)\right\} =\left\{ s,v\right\} \text{ for some }i\in\left[N\right]\text{ and }s\neq v\}.
\]
Thus, the mobility graph has $\left|E\right|=\left|\mathcal{P}\right|-Q$
edges, each of which corresponds to a column of $\bar{\boldsymbol{B}}=\left[\boldsymbol{D},\boldsymbol{R}\right]$
and can be associated with an entry in $\overrightarrow{\Delta}$.

The graph is represented algebraically by the unweighted Laplacian
matrix $\boldsymbol{L}^{u}=\bar{\boldsymbol{B}}\bar{\boldsymbol{B}}'$.
The $j$th diagonal entry of $\boldsymbol{L}^{u}$ gives firm $j$'s
degree in the network: the total number of edges incident to firm
$j$. If firms $j$ and $k$ have an edge between them, then the off-diagonal
entry $\boldsymbol{L}_{jk}^{u}$ will equal -1. Otherwise it will
equal zero.

Figure \ref{fig:mobility=000020graph} depicts a graph with 5 firms
and 7 edges. The unweighted Laplacian of this graph takes the form:

\[
\boldsymbol{L}^{u}=\left[\begin{array}{ccccc}
1 & -1 & 0 & 0 & 0\\
-1 & 4 & -1 & -1 & -1\\
0 & -1 & 3 & -1 & -1\\
0 & -1 & -1 & 3 & -1\\
0 & -1 & -1 & -1 & 3
\end{array}\right].
\]
In the figure, firms are depicted with circles and edges by lines.
The edge names are displayed to the right of each edge. 

I will now introduce some terms of art that are useful for describing
networks. A walk is a sequence of edges that join a set of firms,
a trail is a walk with no repeated edges, and a path is a trail with
no repeated firms. Hence, the sequence $\left\{ e_{23},e_{35},e_{25}\right\} $
is a trail that is not a path.

A graph is connected if there is a path from any firm to any other
firm. The edge connectivity $\lambda\left(G\right)$ of a graph is
the number of edges that need to be removed from the graph for it
to become disconnected. Formally,
\[
\lambda\left(G\right)=\min\Bigl\{\left|S\right|:S\subseteq E,\;\text{and }(V,E\setminus S)\text{ is disconnected}\Bigr\}.
\]
The graph depicted in Figure \ref{fig:mobility=000020graph} has $\lambda\left(G\right)=1$
because removing edge $e_{12}$ severs firm 1 from the network. 

A graph is $k$-edge-connected if $\lambda\left(G\right)\geq k$.
Equivalently, a $k$-edge-connected graph has at least $k$ edge-disjoint
paths between any two vertices \citep{menger1927allgemeinen}. A $k$-edge-connected
component ($k$-ECC) is a maximal subgraph -- the largest collection
of vertices and edges in the original graph -- that is $k$-edge-connected.
Figure \ref{fig:mobility=000020graph} has a single 3-ECC, the vertices
of which are shaded red. Removing any two edges from the 3-ECC fails
to disconnect the graph, while removing the edges $\left\{ e_{23},e_{24},e_{25}\right\} $
partitions the network into two separate connected components. Hence,
this 3-ECC is not also a 4-ECC.

A tree is a connected graph for which there is a unique path between
any pair of firms. A spanning tree is any subset of a connected graph
that contains all firms and is a tree. Kirchhoff's matrix tree theorem
states that any cofactor of the unweighted Laplacian equals the number
of spanning trees in a graph \citep{Kirchhoff1847Aufloesung,spielman2019spectral}.
An implication of this result is that $\boldsymbol{\boldsymbol{L}}_{\left(1\right)}$
is guaranteed to be invertible -- and hence, $\hat{\psi}_{\left(1\right)}$
computable -- whenever the graph is connected. 

The graph depicted in Figure \ref{fig:mobility=000020graph} contains
16 spanning trees. One can use the matrix tree theorem to verify this
claim: computing the absolute value of the determinant of the matrix
formed by deleting any row and column from the unweighted Laplacian
$\boldsymbol{L}^{u}$ will yield 16. However, all 16 of these trees
contain the edge $e_{12}$. Thus, we can only pack a single spanning
tree into this graph without having to share an edge.

The 3-ECC in Figure \ref{fig:mobility=000020graph} also has 16 spanning
trees but we can pack more than one edge-disjoint spanning tree into
this component. Two edge-disjoint trees that span the 3-ECC are depicted
by the blue and red lines. Clearly, we cannot pack a third spanning
tree into this component as the two depicted trees consume all of
the available edges. Evidently, it is not always possible to pack
$k$ edge-disjoint spanning trees into a $k$-ECC. We will return
to this insight in Section \ref{sec:Tree-packing}, where tree packing
is discussed in greater detail.

\section{Trees and branches\protect\label{subsec:Trees-and-branches}}

Returning to the algebraic representation of the fixed effects estimator
in \eqref{eq:undirected}, an important simplification arises when
$\bar{\boldsymbol{B}}$ has dimension $J\times\left(J-1\right)$,
which implies there are only $J-1$ edges connecting the $J$ firms.
An incidence matrix of this form represents a spanning tree. Hence,
$\bar{\boldsymbol{B}}_{\left(1\right)}$ is a square invertible matrix
and $\boldsymbol{L}_{\left(1\right)}^{-1}=\left(\bar{\boldsymbol{B}}_{\left(1\right)}^{\prime}\right)^{-1}\bar{\boldsymbol{W}}^{-1}\bar{\boldsymbol{B}}_{\left(1\right)}^{-1}$.
Plugging this expression into \eqref{eq:undirected} and simplifying
reveals that the firm effects estimator reduces in this case to
\[
\hat{\psi}_{\left(1\right)}=\left(\bar{\boldsymbol{B}}_{\left(1\right)}^{\prime}\right)^{-1}\overrightarrow{\Delta}.
\]
Note that, relative to \eqref{eq:undirected}, the weighting matrix
$\bar{\boldsymbol{W}}$ has disappeared, which reflects that the firm
effects are just-identified in this scenario. This interpretation
is pursued at greater length in \citet{kline2024firm}, where it is
shown that the fundamental restriction of the AKM model is that cycle
sums of the average wage changes $\overrightarrow{\Delta}$ must equal
zero. 

In general, when the mobility network is connected, $\bar{\boldsymbol{B}}$
can be partitioned into submatrices capturing spanning trees and a
set of leftover edges that fail to connect some of the firms. In particular,
if $M$ spanning trees are capable of being packed into the mobility
graph, with corresponding incidence matrices $\boldsymbol{T}_{1},\dots,\boldsymbol{T}_{M}$,
then we can write $\bar{\boldsymbol{B}}=\left[\boldsymbol{T}_{1},\dots,\boldsymbol{T}_{M-1},\boldsymbol{T}_{M}^{+}\right]$,
where $\boldsymbol{T}_{M}^{+}$ appends columns to $\boldsymbol{T}_{M}$
that capture any leftover edges. This set of incidence matrices $\left\{ \boldsymbol{T}_{1},\dots,\boldsymbol{T}_{M-1},\boldsymbol{T}_{M}^{+}\right\} $
define our branches: each branch is a subgraph of the data that connects
all firms in the mobility network.

Removing the first row of $\bar{\boldsymbol{B}}$ yields $\bar{\boldsymbol{B}}_{\left(1\right)}=\left[\boldsymbol{T}_{\left(1\right),1},\dots,\boldsymbol{T}_{\left(1\right),M-1},\boldsymbol{T}_{\left(1\right),M}^{+}\right]$.
It is useful to also partition $\overrightarrow{\Delta}=\left(\overrightarrow{\Delta}_{1}^{\prime},\dots,\overrightarrow{\Delta}_{M}^{\prime}\right)^{\prime}$
and $\bar{\mathbf{W}}=diag\left(\boldsymbol{\bar{W}}_{1},\dots,\bar{\boldsymbol{W}}_{M}\right)$.
We can now define branch-specific estimates
\[
\hat{\psi}_{\left(1\right),b}=\begin{cases}
\left(\boldsymbol{T}_{\left(1\right),b}^{\prime}\right)^{-1}\overrightarrow{\Delta}_{b} & \text{if \ensuremath{b<M}}\\
\left(\boldsymbol{T}_{\left(1\right),M}^{+}\bar{\boldsymbol{W}}_{M}\left(\boldsymbol{T}_{\left(1\right),M}^{+}\right)^{\prime}\right)^{-1}\boldsymbol{T}_{\left(1\right),M}^{+}\bar{\boldsymbol{W}}_{M}\overrightarrow{\Delta}_{M} & \text{if \ensuremath{b=M}}
\end{cases}.
\]
This representation reflects the earlier finding that the weights
are irrelevant for trees. However, the weights may matter for the
last branch $\left(b=M\right)$ because the graph may contain edges
that are not a part of any spanning tree. Including these leftover
edges forms cycles -- disjoint paths between vertices -- on the
subgraph corresponding to the last branch. The inclusion of these
cycles will tend to improve the precision of the last branch relative
to the others. 

With all the edges assigned to a branch, no information in the microdata
is wasted. The following proposition formally establishes that the
full-sample OLS estimator $\hat{\psi}_{\left(1\right)}$ can be written
as a linear combination of the branch-specific estimates $\left\{ \hat{\psi}_{\left(1\right),b}\right\} _{b=1}^{M}$. 
\begin{prop}
\label{Prop:=000020lincom}Suppose the mobility network is connected
and contains $M$ edge-disjoint spanning trees. Then $\hat{\psi}_{\left(1\right)}=\sum_{b=1}^{M}\boldsymbol{C}_{\left(1\right),b}\hat{\psi}_{\left(1\right),b}$,
where each $\boldsymbol{C}_{\left(1\right),b}$ is a $\left(J-1\right)\times\left(J-1\right)$
matrix and $\sum_{b=1}^{M}\boldsymbol{C}_{\left(1\right),b}=\boldsymbol{I}$.
\end{prop}
\begin{proof}
We can write 
\[
\boldsymbol{\boldsymbol{L}}_{\left(1\right)}=\sum_{b=1}^{M-1}\boldsymbol{T}_{\left(1\right),b}\bar{\boldsymbol{W}}_{b}\boldsymbol{T}_{\left(1\right),b}^{\prime}+\boldsymbol{T}_{\left(1\right),M}^{+}\bar{\boldsymbol{W}}_{M}\left(\boldsymbol{T}_{\left(1\right),M}^{+}\right)^{\prime},
\]
\[
\bar{\boldsymbol{B}}_{\left(1\right)}\bar{\boldsymbol{W}}\overrightarrow{\Delta}=\sum_{b=1}^{M-1}\boldsymbol{T}_{\left(1\right),b}\bar{\boldsymbol{W}}_{b}\overrightarrow{\Delta}_{b}+\boldsymbol{T}_{\left(1\right),M}^{+}\bar{\boldsymbol{W}}_{M}\overrightarrow{\Delta}_{M}.
\]
Now define
\[
\boldsymbol{C}_{\left(1\right),b}=\begin{cases}
\boldsymbol{\boldsymbol{L}}_{\left(1\right)}^{-1}\boldsymbol{T}_{\left(1\right),b}\bar{\boldsymbol{W}}_{b}\boldsymbol{T}_{\left(1\right),b}^{\prime} & \text{if \ensuremath{b<M}}\\
\boldsymbol{\boldsymbol{L}}_{\left(1\right)}^{-1}\boldsymbol{T}_{\left(1\right),M}^{+}\bar{\boldsymbol{W}}_{M}\left(\boldsymbol{T}_{\left(1\right),M}^{+}\right)^{\prime} & \text{if \ensuremath{b=M}},
\end{cases}
\]
and note that 
\[
\sum_{b=1}^{M}\boldsymbol{C}_{\left(1\right),b}=\boldsymbol{\boldsymbol{L}}_{\left(1\right)}^{-1}\left(\sum_{b=1}^{M-1}\boldsymbol{T}_{\left(1\right),b}\bar{\boldsymbol{W}}_{b}\boldsymbol{T}_{\left(1\right),b}^{\prime}+\boldsymbol{T}_{\left(1\right),M}^{+}\bar{\boldsymbol{W}}_{M}\left(\boldsymbol{T}_{\left(1\right),M}^{+}\right)^{\prime}\right)=\boldsymbol{\boldsymbol{L}}_{\left(1\right)}^{-1}\boldsymbol{\boldsymbol{L}}_{\left(1\right)}=\boldsymbol{I}.
\]

For each $b<M$, 
\[
\boldsymbol{\boldsymbol{L}}_{\left(1\right)}^{-1}\boldsymbol{T}_{\left(1\right),b}\bar{\boldsymbol{W}}_{b}\overrightarrow{\Delta}_{b}=\left[\boldsymbol{\boldsymbol{L}}_{\left(1\right)}^{-1}\boldsymbol{T}_{\left(1\right),b}\bar{\boldsymbol{W}}_{b}\boldsymbol{T}_{\left(1\right),b}^{\prime}\right]\hat{\psi}_{\left(1\right),b}=\boldsymbol{C}_{\left(1\right),b}\hat{\psi}_{\left(1\right),b},
\]
while, for $b=M$, 
\[
\boldsymbol{\boldsymbol{L}}_{\left(1\right)}^{-1}\boldsymbol{T}_{\left(1\right),M}^{+}\bar{\boldsymbol{W}}_{M}\overrightarrow{\Delta}_{M}=\left[\boldsymbol{\boldsymbol{L}}_{\left(1\right)}^{-1}\boldsymbol{T}_{\left(1\right),M}^{+}\bar{\boldsymbol{W}}_{M}\left(\boldsymbol{T}_{\left(1\right),M}^{+}\right)^{\prime}\right]\hat{\psi}_{\left(1\right),M}=\boldsymbol{C}_{\left(1\right),M}\hat{\psi}_{\left(1\right),M}.
\]
Hence,
\[
\hat{\psi}_{\left(1\right)}=\boldsymbol{\boldsymbol{L}}_{\left(1\right)}^{-1}\bar{\boldsymbol{B}}_{\left(1\right)}\bar{\boldsymbol{W}}\overrightarrow{\Delta}=\sum_{b=1}^{M}\boldsymbol{C}_{\left(1\right),b}\hat{\psi}_{\left(1\right),b}.
\]
\end{proof}
The matrix $\boldsymbol{C}_{\left(1\right),b}$ can be shown to capture
the relative precision of the branch-specific estimates relative to
the full-sample estimate when the wage errors are homoscedastic \citep{jochmans2019fixed}.
Thus, $\hat{\psi}_{\left(1\right)}$ can be thought of as a precision-matrix-weighted
average of the branch estimates, albeit with the potential for negative
elementwise weights.

Now letting
\[
\hat{\psi}=\left(0,\hat{\psi}_{\left(1\right)}^{\prime}\right)^{\prime},\ \hat{\psi}_{b}=\left(0,\hat{\psi}_{\left(1\right)b}^{\prime}\right)^{\prime},\ \boldsymbol{C}_{b}=diag\left(0,\boldsymbol{C}_{\left(1\right),b}\right),\ \text{and }\hat{\phi}_{b}=\boldsymbol{C}_{b}\hat{\psi}_{b},
\]
Proposition \ref{Prop:=000020lincom} yields the additive decomposition
introduced in \eqref{eq:decomp}:
\[
\hat{\psi}=\sum_{b=1}^{M}\hat{\phi}_{b}.
\]
This representation reveals that each $\hat{\phi}_{b}$ measures the
influence of a branch on the full-sample estimator. In fact, one can
show that each $\hat{\phi}_{b}$ gives the branch sum of edge-level
contributions to the recentered influence function of the full-sample
estimator. Assumption \ref{Pair-independence} guarantees that the
$\left\{ \hat{\phi}_{b}\right\} _{b=1}^{M}$ are mutually independent,
providing a transparent foundation for assessing uncertainty in $\hat{\psi}$. 

Computation of $\hat{\phi}_{b}=\left(0,\hat{\phi}_{\left(1\right),b}^{\prime}\right)^{\prime}$
is greatly aided by the observation that
\[
\hat{\phi}_{\left(1\right),b}=\begin{cases}
\boldsymbol{\boldsymbol{L}}_{\left(1\right)}^{-1}\boldsymbol{T}_{\left(1\right),b}\bar{\boldsymbol{W}}_{b}\overrightarrow{\Delta}_{b} & \text{if \ensuremath{b<M}}\\
\boldsymbol{\boldsymbol{L}}_{\left(1\right)}^{-1}\boldsymbol{T}_{\left(1\right),M}^{+}\bar{\boldsymbol{W}}_{M}\overrightarrow{\Delta}_{M} & \text{if \ensuremath{b=M}.}
\end{cases}
\]
Each of these expressions is simply a least squares regression that
is no more computationally expensive to solve than the weighted least
squares problem in \eqref{eq:undirected}. Likewise, letting $\overrightarrow{\Delta}_{-b}=\boldsymbol{T}_{\left(1\right),b}^{\prime}\hat{\psi}_{\left(1\right),-b}$
for $b<M$ and $\overrightarrow{\Delta}_{-M}=\left(\boldsymbol{T}_{\left(1\right),M}^{+}\right)^{\prime}\hat{\psi}_{\left(1\right),-M}$,
the leave-out contributions can be written
\[
\hat{\phi}_{\left(1\right),-b}=\begin{cases}
\boldsymbol{\boldsymbol{L}}_{\left(1\right)}^{-1}\boldsymbol{T}_{\left(1\right),b}\bar{\boldsymbol{W}}_{b}\overrightarrow{\Delta}_{-b} & \text{if \ensuremath{b<M}}\\
\boldsymbol{\boldsymbol{L}}_{\left(1\right)}^{-1}\boldsymbol{T}_{\left(1\right),M}^{+}\bar{\boldsymbol{W}}_{M}\overrightarrow{\Delta}_{-M} & \text{if \ensuremath{b=M},}
\end{cases}
\]
which is again a least squares regression, where now the dependent
variable is the predicted oriented wage change based on the leave-branch-out
firm effect estimates.

\section{Building branches\protect\label{sec:Tree-packing}}

The key challenge in constructing branch-specific estimates is building
the incidence matrices $\boldsymbol{T}_{1},\dots,\boldsymbol{T}_{M}$
corresponding to the graph's spanning trees. Doing so requires first
determining how many spanning trees can be packed into the graph.
The following result due to \citet{nash1961edge} and \citet{tutte1961decomposing}
provides a useful foundation for answering this question.
\begin{thm*}[Nash-Williams--Tutte, 1961]
A graph $G=\left(V,E\right)$ can pack $M$ edge-disjoint spanning
trees if and only if, for every partition $\left\{ V_{1},\dots,V_{r}\right\} $
of $V$ into $r\geq2$ blocks, the number of edges connecting vertices
in different blocks is at least $M\left(r-1\right)$.
\end{thm*}
Intuitively, a spanning tree must cross any partition of the vertices
to span the graph. The Nash-Williams--Tutte theorem establishes that
if there is a partition of the vertices that yields fewer than $M\left(r-1\right)$
crossings, then there must be fewer than $M$ spanning trees capable
of being packed into the graph. Conversely, if $M$ spanning trees
can be packed into the graph, then the number of edges crossing any
partition must be at least $M\left(r-1\right)$.

The packing number of a graph, $\tau\left(G\right)$, gives the maximum
number of edge-disjoint spanning trees in $G$. The following corollary
of the Nash-Williams--Tutte theorem provides us with bounds on $\tau\left(G\right)$
in terms of the graph's edge connectivity $\lambda\left(G\right)$.
\begin{cor*}
For a connected graph $G=\left(V,E\right)$, 
\[
\left\lfloor \lambda\left(G\right)/2\right\rfloor \leq\tau\left(G\right)\leq\lambda\left(G\right),
\]
where $\left\lfloor \cdot\right\rfloor $ denotes the floor operator.
\end{cor*}
A proof of this corollary can be found in \citet{kundu1974bounds}.
The upper bound $\tau\left(G\right)\leq\lambda\left(G\right)$ follows
by taking $r=2$ and choosing a partition that yields exactly $\lambda\left(G\right)$
crossings: any spanning tree must include at least one such crossing
edge to connect the two blocks, and edge-disjoint spanning trees must
use distinct crossing edges. The lower bound $\tau(G)\ge\lfloor\lambda(G)/2\rfloor$
comes from a double-counting of edges across any $r$-way partition:
each block $V_{i}$ has at least $\lambda(G)$ edges with one endpoint
in $V_{i}$ and the other in $V\setminus V_{i}$; summing over $i=1,\dots,r$
counts each crossing edge twice, so the total number of crossings
is at least $r\lambda(G)/2$. In the case of Figure \ref{fig:mobility=000020graph},
this corollary tells us that $\tau\left(G\right)\in\left[1,3\right]$
in the graph's 3-ECC. As the Figure illustrates, there are exactly
two edge-disjoint spanning trees in that component.

In the prototypical worker--firm dataset, the edge connectivity is
zero: there are usually multiple connected components of firms. Since
the work of \citet{abowd2002computing}, the convention in the literature
has been to focus on the largest connected component of the mobility
graph (i.e., the largest $1$-ECC). Typically many firms in the largest
$1$-ECC are connected by only a single edge, in which case no more
than one edge-disjoint spanning tree can be packed into the graph.
The Nash-Williams--Tutte theorem suggests a reasonable way to find
multiple disjoint spanning trees is to narrow the scope of investigation
from the largest $1$-ECC to the largest $k$-ECC for $k>1$. Doing
so requires first finding the largest $k$-ECC and then packing this
subgraph. 

\subsection{Finding the largest $k$-ECC}

A key tool that facilitates finding a $k$-ECC is the Gomory-Hu tree
\citep{GomoryHu1961}, which is a weighted tree on $V$. The edge
weights of this tree provide the number of edge removals required
to disconnect a pair of vertices. Figure \ref{fig:Gomory-Hu} provides
the Gomory-Hu tree associated with the graph in Figure \ref{fig:mobility=000020graph}.
The edge connecting Firm 1 to Firm 2 has a weight of 1, reflecting
that Firm 1 can be disconnected from the network by removing $e_{12}$.
In contrast, disconnecting Firm 2 from Firm 3 would require removing
three edges $\left\{ e_{23},e_{24},e_{25}\right\} $, which is why
the Gomory-Hu edge between these vertices has a weight of 3.

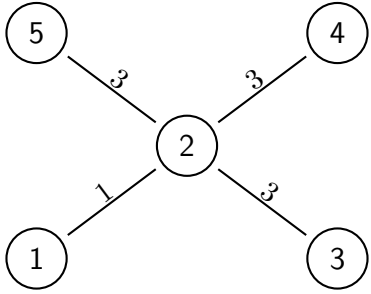
\begin{figure}[H]
\centering
\centering{}% Gomory–Hu tree for the 5-node network from the Laplacian
\begin{tikzpicture}[scale=1.0,
    node_style/.style={circle, draw=black, thick, minimum size=8mm, font=\sffamily},
    gh_edge/.style={thick, black, shorten >=3pt, shorten <=3pt},
    weight_label/.style={font=\small, inner sep=1pt, sloped, above}
]

% Central node
\node[node_style] (2) at (0,0)      {2};
% Leaves
\node[node_style] (1) at (-2,-1.5)  {1};
\node[node_style] (3) at (2,-1.5)   {3};
\node[node_style] (4) at (2,1.5)    {4};
\node[node_style] (5) at (-2,1.5)   {5};

% Gomory–Hu edges with min-cut weights
\draw[gh_edge] (2) -- node[weight_label] {1} (1);
\draw[gh_edge] (2) -- node[weight_label] {3} (3);
\draw[gh_edge] (2) -- node[weight_label] {3} (4);
\draw[gh_edge] (2) -- node[weight_label] {3} (5);

\end{tikzpicture}
\caption{Gomory-Hu Tree of graph in Figure \ref{fig:mobility=000020graph}\protect\label{fig:Gomory-Hu}}
\end{figure}

A useful feature of Gomory-Hu trees is that the minimal edge weight
on the path between any pair of vertices $\left(s,v\right)\in V\times V$
gives the number of edge removals required to disconnect those vertices.
As a result, dropping all edges with weight less than $k$ reveals
the set of all $k$-ECCs. 

For example, if we threshold the tree in Figure \ref{fig:Gomory-Hu}
at $k=2$, we are left with a single 2-ECC formed by the vertex set
$\left\{ 2,3,4,5\right\} $. Several algorithms exist for rapidly
building the Gomory-Hu tree associated with a graph. I will use an
algorithm due to \citet{Gusfield1990}, implemented in the \texttt{igraph}
package \citep{csardi2023igraph}, which has worst-case computational
complexity of $O\left(J^{4}\right)$.\footnote{For very large graphs, building the full Gomory-Hu tree will be excessively
costly. In those settings, one can work with an approach due to \citet{nagamochi1992linear,nagamochi1992computing}
that avoids computing the entire tree and directly finds the $k$-ECC.} Once we have pruned the Gomory-Hu tree, the $k$-ECC with the most
vertices can be selected.

\subsection{Packing a $k$-ECC}

After the largest $k$-ECC has been found, we would like to extract
as many edge-disjoint spanning trees from the graph as possible. This
is known as the \textit{tree-packing problem}. Formally, the tree
packing problem is an integer linear programming (ILP) problem that
can be expressed as follows:
\[
\max_{x}\ \sum_{T\in\mathcal{T}}x_{T}\quad\text{s.t.}\quad\sum_{T\ni e}x_{T}\le1,\ \forall e\in E,\qquad x_{T}\in\{0,1\},\ \forall T\in\mathcal{T},
\]
where $\mathcal{T}$ is the collection of all spanning trees in the
graph and $x_{T}$ is an indicator for whether tree $T\in\mathcal{T}$
is chosen. The matrix tree theorem tells us that $\left|\mathcal{T}\right|$
can grow exponentially with the scale of the graph, making direct
enumeration infeasible. 

Rather than tackle the ILP directly, I rely on an iterative approach
proposed by \citet{roskind_tarjan_1985}. Their algorithm requires
pre-specifying the number $M$ of edge-disjoint spanning trees that
one is seeking to extract from a $k$-ECC. When $M\leq\tau\left(G\right)$,
the procedure is guaranteed to find exactly $M$ trees. In contrast,
greedily extracting trees sequentially (e.g., via Kruskal's algorithm)
may find fewer than $\tau\left(G\right)$ trees. An example is provided
in \ref{sec:Appendix:-An-example}.

The unweighted Roskind-Tarjan algorithm, which is implemented in SageMath,
is well suited to large graphs, exhibiting worst-case computational
complexity $O\left(J^{2}M^{2}\right)$. If we apply the algorithm
to a $k$-ECC, then the Corollary assures us that we will be able
to find at least $\left\lfloor k/2\right\rfloor $ trees. 

\subsection{A Prune-and-Pack algorithm}

Algorithm \ref{alg:Pack=000020k-ECC} provides a sketch of the composite
procedure for pruning to the $k$-ECC and extracting the edge-disjoint
spanning trees.

\begin{algorithm}[h]
\caption{Pack $k$-ECC\protect\label{alg:Pack=000020k-ECC}}

Fix a $k\in\mathbb{N}$

Drop all vertices with degree \textless$k$. Drop all but largest
connected component. Repeat until no more vertices dropped.

Build a Gomory-Hu tree from the $k$-core in the previous step. Remove
all tree edges with weight $<k$, then take the largest connected
component.

Initialize $M=\left\lfloor k/2\right\rfloor $. While Roskind-Tarjan
returns $M$ disjoint spanning trees, increment $M$, stopping when
no more spanning trees can be found or when $M=k$.
\end{algorithm}
Step 2 of this algorithm is a screen that exploits the fact that a
vertex with degree $<k$ cannot be a member of the $k$-ECC. The maximal
subgraph of a network in which all vertices have degree $\geq k$
is known as the $k$-core, which need not be connected. Step 3 searches
within the $k$-core for the largest $k$-ECC. In step 4, we first
try $M=\left\lfloor k/2\right\rfloor $, and then explore higher values
of $M$ until no more spanning trees can be found. The algorithm assumes
the $k$-ECC is not also a $\left(k+1\right)$-ECC, which implies
$\tau\left(G\right)\leq k$. Hence, the largest possible value of
$M$ explored in Step 4 is $k$. 

\subsection{Re-packing for reproducibility}

While Algorithm \ref{alg:Pack=000020k-ECC} is guaranteed to deliver
as many trees as can be packed into the $k$-ECC of interest, the
packing it produces is not unique: the Roskind-Tarjan step relies
on the (ultimately arbitrary) labeling of the vertex identifiers.
Rather than rely on an arbitrary packing, one can consider several
random packings and average downstream branch-based estimates over
them. Below, I explore alternate packings by randomly reshuffling
the vertex ids of the graph representation and reapplying the Roskind-Tarjan
algorithm 100 times. Note that, in contrast to random sample-splitting
of individual observations, re-packing does not alter the estimand.

While averaging over random packings aids reproducibility \citep{ritzwoller2023reproducible},
it also adds to the burden of releasing estimates to the public. For
outside researchers to take full advantage of $P$ packings, $P$
copies of the $\left\{ \hat{\psi}_{b},\hat{\phi}_{b},\hat{\phi}_{-b}\right\} _{b=1}^{M}$
must be hosted, each derived from a different vertex ordering. Thus,
considering 10 alternate packings increases storage requirements by
an order of magnitude. We will see below that some statistics exhibit
little variability across alternate packings, suggesting that very
few packings may suffice for some purposes.

\section{Empirical application\protect\label{sec:Application}}

I now turn to analyzing an extract of the benchmark Veneto dataset
that was also studied in \citet{kline2020leave} and \citet{kline2024firm}.
The extract consists of 1,859,459 person-year observations from the
years 1999 and 2001. The largest connected component contains 73,933
firms and 747,205 workers, 197,572 of whom switch employers between
the two years. The worker moves between these two years involve 150,417
ordered origin-destination pairs $\left|\mathcal{P}\right|$. Exactly
1,500 pairs of firms have flows in both directions. Hence, the mobility
network formed by these pairs contains $\left|E\right|=$150,417-1,500=148,917
undirected edges.

Note that the average number of movers per edge is quite low as 197,572/148,917$\approx$1.3.
Thus, branching firm effect estimates based upon undirected edges
is likely to provide a reasonable approximation to what is possible
when treating each individual worker move as an edge. This is fortunate,
as packing spanning trees into multi-edge graphs turns out to be significantly
more complex than packing them into simple graphs \citep{barahona1995packing}. 

Table \ref{tab:Network-properties} reports the results of pruning
the mobility graph to the largest $k$-ECC for different choices of
$k$ as in Algorithm \ref{alg:Pack=000020k-ECC}. 

\begin{table}[H]
\begin{centering}
\begin{tabular}{|c|c|c|c|c|c|c|}
\hline 
$k$ & 1 & 2 & 3 & 4 & 5 & 6\tabularnewline
\hline 
\hline 
Firms in $k$-core & 73,933 & 41,093 & 21,570 & 11,145 & 5,682 & 3,128\tabularnewline
\hline 
Firms in $k$-ECC ($J$) & 73,933 & 41,054 & 21,565 & 11,145 & 5,682 & 3,128\tabularnewline
\hline 
Edges in $k$-ECC ($\left|E\right|$) & 148,917 & 116,026 & 80,561 & 51,824 & 31,677 & 19,796\tabularnewline
\hline 
Movers in $k$-ECC ($N$) & 197,572 & 158,149 & 114,717 & 78,908 & 53,131 & 36,903\tabularnewline
\hline 
Spanning Trees ($M$) & 1 & 1 & 2 & 3 & 3 & 4\tabularnewline
\hline 
\end{tabular}
\par\end{centering}
\caption{Network properties of $k$-ECCs \protect\label{tab:Network-properties}}
\end{table}
The size of the $k$-core contracts rapidly with $k$, reflecting
that most firms have low degree. To some extent this phenomenon is
an artifact of studying mobility over only a 2-year horizon. However,
adding further years of mobility data may introduce new firms that
only contain a single edge, which could actually lower the average
degree of the network. An interesting approach for future work would
be to fix a set of firms of interest in a base year and fill in the
edges between them with additional years of mobility data. 

The sparse nature of the Veneto network makes the $k$-core a very
close approximation to the largest $k$-ECC. While this finding need
not generalize to larger, or more densely connected, networks, these
patterns suggest the $k$-core is likely to provide a useful starting
guess for the $k$-ECC in larger graphs where full computation of
a Gomory-Hu tree would be impractical. In truly massive graphs, it
may be advisable to simply pack spanning trees into the $k$-core
via the Roskind-Tarjan algorithm rather than refining to the largest
$k$-ECC.

In the Veneto data, pruning the sample to the $2$-ECC leads to the
loss of about 44\% of the firms but fails to yield a second spanning
tree. The Nash-Williams-Tutte theorem suggests we should not be surprised
by this finding as a graph with $\lambda\left(G\right)=2$ can have
at most 2 spanning trees. The $3$-ECC, which retains about 29\% of
the firms but 58\% of the movers in the original graph, does contain
a second spanning tree. A third tree is found in the $4$-ECC, which
has 15\% of the firms and 39\% of the movers in the original graph.
A fourth tree emerges with the $6$-ECC, which contains only 4\% of
firms but 17\% of the movers in the original graph. 

It is interesting to compare these sample dimensions to what one would
obtain by randomly splitting the sample at the mover level. To build
this benchmark, I randomly assign each worker to one of $M$ splits
with equal probability, varying $M$ from one to four. In each split,
I calculate the largest connected component and then intersect the
vertex set of those components across splits to arrive at the number
of firms for which $M$ independent unbiased estimates can be computed.
Table \ref{tab:Random-splits} reports the results of repeating this
process 500 times. 

\begin{table}[H]
\begin{centering}
\begin{tabular}{|c|c|c|c|c|}
\hline 
Splits ($M$) & 1 & 2 & 3 & 4\tabularnewline
\hline 
\hline 
\textbf{Number of firms} &  &  &  & \tabularnewline
\hline 
25th Percentile & 73,933 & 28,680 & 13,034 & 6,537\tabularnewline
\hline 
Median & 73,933 & 28,737 & 13,077 & 6,576\tabularnewline
\hline 
75th Percentile & 73,933 & 28,797 & 13,127 & 6,607\tabularnewline
\hline 
\textbf{Overlap across simulations} &  &  &  & \tabularnewline
\hline 
25th Percentile & 73,933 & 22,745 & 9,328 & 4,482\tabularnewline
\hline 
Median & 73,933 & 22,804 & 9,366 & 4,509\tabularnewline
\hline 
75th Percentile & 73,933 & 22,861 & 9,406 & 4,536\tabularnewline
\hline 
\end{tabular}
\par\end{centering}
\caption{Firm effects estimable in each of $M$ random splits (500 simulations)\protect\label{tab:Random-splits}}
\end{table}

Splitting the sample in half yields a pair of independent estimates
for roughly one third of the firms. A three-way split yields three
estimates for roughly 18\% of firms. A four-way split yields four
estimates for about 9\% of firms. While little variability emerges
across the 500 simulation draws in the number of firms for which $M$
estimates can be computed, the composition of these sets of firms
varies considerably. The bottom panel of the table reports quantiles
of the degree of overlap in the sets of firms for which $M$ estimates
can be computed across all $\binom{500}{2}$ pairs of simulations.
The table reveals that reshuffling a split into $M=2$ groups would
yield less than 22,804 firms in common with the previous split half
of the time. Likewise, re-randomizing a three-way split would yield
less than 9,366 of the firms present in a previous random split half
of the time.

Comparing these findings to Table \ref{tab:Network-properties} reveals
that randomly splitting workers into $M$ groups yields more firms
with $M$ estimates than does pruning to the largest $k$-ECC with
$M$ trees. For example, the largest $3$-ECC has 21,565 firms, while
randomly assigning workers to one of two splits yields a median of
28,737 firms with two independent estimates. This discrepancy suggests
that there exist subgraphs of the largest $2$-ECC that strictly nest
the largest $3$-ECC yet contain two edge-disjoint spanning trees.
How to systematically explore the space of such intermediate subgraphs
is an interesting question. For example, one could potentially append
pruned edges to the branches of the $3$-ECC, growing the size of
the network in a way that ensures each branch connects the same set
of vertices. I leave such investigations to future research. 

While random splitting delivers independent estimates for more firms,
a major advantage of the pruning and packing strategy is that the
target population is deterministic. Another advantage is that the
full sample estimator decomposes linearly into $M$ branch-specific
estimates, which facilitates uncertainty quantification. In contrast,
the random splitting approach yields a full-sample estimator that
includes a fixed effect for every firm in the \textit{union} of the
$M$ connected sets. While some of these parameters have $M$ independent
estimates, others may have only a single estimate, which hinders uncertainty
quantification even when ignoring the randomness in the split itself. 

The next subsections illustrate how branches constructed from packing
$k$-ECC's can be put to work. In what follows, I compute the full-sample
estimator $\hat{\psi}$ for each $k$-ECC using the adjusted wage
changes $\hat{\Delta}$ between origin-destination pairs described
in \citet{kline2024firm}, which only adjust for a year fixed effect.
Note that each of the 197,572 worker moves contribute to at most one
of these wage changes. To facilitate computation, the vector $\hat{\Delta}$
is collapsed to a shorter vector of oriented wage changes $\overrightarrow{\Delta}$,
which Assumption \ref{Pair-independence} implies are mutually independent.
The $\overrightarrow{\Delta}$ are then used to construct the branch-specific
estimates $\hat{\psi}_{b}$ and $\hat{\phi}_{b}$. While the methods
described in the next subsection only apply to branches, the approaches
in Sections \ref{subsec:Moment-estimation} and \ref{subsec:Shrinkage}
can just as easily be applied to estimates constructed from random
splits.

\subsection{Quantifying uncertainty}

The arguments of Section \ref{sec:Basic-idea} suggest estimating
the variance matrix $\Sigma$ of the vector $\hat{\psi}$ of full-sample
firm effect estimates with $\hat{\Sigma}.$ To illustrate the potential
usefulness of this matrix, I consider the projection of $\psi$ onto
the matrix $\boldsymbol{X}=\left[\boldsymbol{1},\ln f\right]$, which
consists of an intercept and the logarithm of average firm size across
the two time periods. Letting $\iota=\left[0,1\right]'$, the quantity
$\gamma=\iota'\boldsymbol{S}_{xx}^{-1}\boldsymbol{X}'\psi$ approximates
the elasticity of firm wage effects with respect to firm size. 

As documented in \citet{bloom2018disappearing} and \citet{kline2024firm}
firm wage effects do not seem to be linear in log firm size, instead
exhibiting a concave relationship in which the largest firms exhibit
wages effects roughly equivalent to or lower than their slightly smaller
peers. Consequently, this linear projection will not coincide exactly
with the conditional expectation function. Nonetheless, the plug-in
estimator $\hat{\gamma}=\iota'\boldsymbol{S}_{xx}^{-1}\boldsymbol{X}'\hat{\psi}$
remains unbiased for the projection coefficient $\gamma$ conditional
on the covariates $\boldsymbol{X}$, which we treat as fixed. 

A branch-based estimator of the variance of the projection coefficient
is 
\[
\hat{\mathbb{V}}_{branch}\left[\hat{\gamma}\right]=\max\left\{ 0,\iota'\boldsymbol{S}_{xx}^{-1}\boldsymbol{X}'\hat{\Sigma}\boldsymbol{X}\boldsymbol{S}_{xx}^{-1}\iota\right\} ,
\]
where the $\max$ operator accounts for the fact that $\hat{\Sigma}$
need not be positive semi-definite. It is instructive to compare this
estimator to the naive variance estimate that comes from treating
the entries of $\hat{\psi}$ as mutually independent in a second-step
regression:
\[
\hat{\mathbb{V}}_{HC0}\left[\hat{\gamma}\right]=\iota'\boldsymbol{S}_{xx}^{-1}\boldsymbol{X}'diag\left\{ \left(\hat{\psi}-\boldsymbol{X}\boldsymbol{S}_{xx}^{-1}\boldsymbol{X}'\hat{\psi}\right)^{\circ2}\right\} \boldsymbol{X}\boldsymbol{S}_{xx}^{-1}\iota.
\]
While $\hat{\mathbb{V}}_{HC0}\left[\hat{\gamma}\right]$ is robust
to misspecification of the conditional expectation function, it neglects
dependence between the estimated firm effects, which can lead to large
biases in either direction.

Table \eqref{tab:Elasticity-of-firm} reports the estimated projection
coefficients, the two sets of standard errors, and the mean firm size
in each $k$-ECC. As $k$ grows larger, the sample is restricted to
larger firms, which tend to have greater degree in the mobility network.
The estimated projection coefficient $\hat{\gamma}$ declines substantially
with firm size, reflecting that the relationship between firm wage
effects and log firm size is concave. 

\begin{table}[H]
\begin{centering}
\begin{tabular}{|l|c|c|c|c|c|c|}
\hline 
\qquad{}\qquad{}\qquad{}\qquad{}$k$ & 1 & 2 & 3 & 4 & 5 & 6\tabularnewline
\hline 
\hline 
Full-sample estimate ($\hat{\gamma}$) & 0.0446 & 0.0329 & 0.0235 & 0.0199 & 0.0191 & 0.0209\tabularnewline
\hline 
Standard error ($\hat{\mathbb{V}}_{branch}\left[\hat{\gamma}\right]^{1/2}$) &  &  &  &  &  & \tabularnewline
\hline 
\qquad{}Mean & - & - & 0.0014 & 0.0004 & 0.0047 & 0.0001\tabularnewline
\hline 
\qquad{}Std dev & - & - & \{0.0013\} & \{0.0007\} & \{0.0004\} & \{0.0006\}\tabularnewline
\hline 
Naive std err ($\hat{\mathbb{V}}_{HC0}\left[\hat{\gamma}\right]^{1/2}$) & 0.0009 & 0.0010 & 0.0013 & 0.0018 & 0.0023 & 0.0032\tabularnewline
\hline 
Mean firm size & 13.6 & 21.6 & 34.6 & 55.1 & 88.0 & 131.6\tabularnewline
\hline 
\end{tabular}
\par\end{centering}
\caption{Elasticity of firm wage effects with respect to firm size\protect\label{tab:Elasticity-of-firm}}
{\footnotesize\bigskip{}
}{\footnotesize\par}

{\footnotesize Notes: the numbers reported in the row labeled ``Mean''
are square roots of the averages of $\hat{\mathbb{V}}_{branch}\left[\hat{\gamma}\right]$
over 100 randomly drawn packings. The numbers in curly braces give
the standard deviation of the standard error across these packings.
This calculation relies on the delta method: it is the standard deviation
across packings of $\hat{\mathbb{V}}_{branch}\left[\hat{\gamma}\right]$
scaled by twice the square root of its mean across packings. Dividing
this number by 10 gives the standard error associated with repacking
variability.}{\footnotesize\par}
\end{table}

The reported value of the branch-based standard error is constructed
by first averaging $\hat{\mathbb{V}}_{branch}\left[\hat{\gamma}\right]$
over $P=100$ alternate packings and then taking the square root.
The standard deviation of this standard error across alternate packings
is reported in curly braces. In most cases, this standard deviation
turns out to be on the order of $10^{-4}$. Hence, two researchers,
each relying on a single randomly selected packing, would be expected
to arrive at branch-based standard errors within $10^{-4}$ of each
other. Averaging across 100 packings reduces the packing noise level
to the order of $10^{-5}$. 

The greatest packing uncertainty arises in the 3-ECC, where the square
root of the average of $\hat{\mathbb{V}}_{branch}\left[\hat{\gamma}\right]$
over 100 packings is 0.0014. The packing standard error of this standard
error estimate is $0.0013/10$. Hence, an approximate 95\% confidence
interval for the standard error that would be achieved by averaging
over the universe of possible packings is $\left[0.0011,0.0017\right]$.
Note that this interval includes the naive standard error of 0.0013.
In practice, both standard errors are extremely small, implying $z$-scores
in the range 17-18. 

For $k>3$, the packing uncertainty is negligible after averaging
across the 100 packings, indicating that the branch-based standard
errors that would result from averaging over the universe of packings
likely differ from the naive standard errors. Substantively, however,
these differences turn out to be small: the branch-based standard
errors, like their naive counterparts, all indicate that the firm-size
wage elasticity is precisely estimated.

While the order of magnitude jumps across $k$-ECCs in the estimated
branch-based standard errors are plausibly attributable to chance
(i.e., to the errors $\boldsymbol{u}$), these patterns could also
reflect misspecification of the AKM model. As shown in \citet{kline2024firm},
the wage changes $\hat{\Delta}$ found in the Veneto data do not perfectly
obey the restrictions of the model. Suppose that each branch estimator
$\hat{\psi}_{b}$ is unbiased for a different target parameter $\psi_{b}$
and $\mathbb{E}\left[\hat{\psi}_{-b}\right]=\psi_{-b}\neq\psi_{b}$.
In this case,
\[
\mathbb{E}\left[\left(\hat{\phi}_{b}-\hat{\phi}_{-b}\right)\hat{\phi}_{b}'\right]=\boldsymbol{C}_{b}\left(\psi_{b}-\psi_{-b}\right)\phi_{b}'+\mathbb{V}\left[\hat{\phi}_{b}\right].
\]
It follows that the variance estimator is biased:
\[
\mathbb{E}\left[\hat{\Sigma}\right]=\mathbb{V}\left[\hat{\psi}\right]+\underbrace{\frac{1}{2}\sum_{b=1}^{M}\left[\boldsymbol{C}_{b}\left(\psi_{b}-\psi_{-b}\right)\phi_{b}'+\phi_{b}\left(\psi_{b}-\psi_{-b}\right)'\boldsymbol{C}_{b}'\right]}_{\text{estimand instability}}.
\]
While the first term captures the sampling variability of $\hat{\psi}$,
the second term measures estimand instability across branches. Unfortunately,
this term cannot be signed. I leave further study of the properties
of branch-based estimators under misspecification to future work.

\subsection{Moment estimation\protect\label{subsec:Moment-estimation}}

Moments of firm effects are common summary statistics that have long
played an important role in the literature. Since the firm effects
$\psi$ are only identified up to a location shift, it is traditional
to report centered moments as in \citet{abowd1999high}. Researchers
typically weight these central moments by activity measures such as
firm size.

Let $\omega\in\mathbb{R}_{\geq0}^{J}$ be a vector of firm weights
that sums to one $\left(\boldsymbol{1}'\omega=1\right)$. For any
$\ell\geq2$, we can write the plug-in estimator of the $\ell$th
central weighted moment of firm effects as
\[
\hat{\mu}_{\ell,PI}=\omega'\left(\hat{\psi}-\omega'\hat{\psi}\boldsymbol{1}\right)^{\circ\ell}.
\]
Generalizing the approach suggested in Section \ref{sec:Basic-idea},
I use branches to construct the corresponding unbiased estimator:
\[
\hat{\mu}_{\ell}=\binom{M}{\ell}^{-1}\sum_{(b_{1},\dots,b_{\ell})\in\mathcal{B}}\omega'\left(\hat{\psi}_{b_{1}}-\omega'\hat{\psi}_{b_{1}}\boldsymbol{1}\right)\odot\dots\odot\left(\hat{\psi}_{b_{\ell}}-\omega'\hat{\psi}_{b_{\ell}}\boldsymbol{1}\right),
\]
where $\mathcal{B}=\bigl\{(b_{1},\dots,b_{\ell}):1\le b_{1}<\cdots<b_{\ell}\le M\bigr\}$
is the set of all combinations of $\ell$ distinct branches. Note
that, when $M=2$, $\hat{\mu}_{2}$ is the covariance estimator used
in split sample approaches. When $M>2$, $\hat{\mu}_{2}$ averages
across covariance estimates formed by all possible sample splits.
When $\ell>2$, higher-order products are taken.

Table \ref{tab:Centered=000020Moments} reports plug-in and branch-based
estimates of moments of the firm effects weighting by total firm size
across the two periods. The branch-based estimates are again averaged
over 100 packings, with across-packing standard deviations reported
in curly braces. The estimated second moments exhibit little variability
across packings. In contrast, branch-based estimates of third and
fourth moments are quite variable, indicating that averaging across
multiple packings is essential for reproducibility. Fortunately, averaging
across 100 packings seems sufficient to minimize the influence of
the chosen packings on most of the estimates. For example, in the
6-ECC, a 95\% confidence interval for the fourth moment estimate that
would be recovered by averaging over an infinite number of packings
is $0.0036\pm1.96\left(0.0025\right)/10=\left[0.0031,0.0041\right]$.

\begin{table}[H]
\begin{centering}
\begin{tabular}{|l|c|c|c|c|c|c|}
\hline 
\qquad{}\qquad{}$k$ & 1 & 2 & 3 & 4 & 5 & 6\tabularnewline
\hline 
\hline 
Second moment ($\mu_{2}$) &  &  &  &  &  & \tabularnewline
\hline 
\qquad{}Plug-in ($\hat{\mu}_{2,PI}$) & 0.0490 & 0.0381 & 0.0339 & 0.0322 & 0.0307 & 0.0308 \tabularnewline
\hline 
\qquad{}Branch estimate ($\hat{\mu}_{2}$) &  &  &  &  &  & \tabularnewline
\hline 
\qquad{}\qquad{}Mean & - & - & 0.0276 & 0.0244 & 0.0306 & 0.0273\tabularnewline
\hline 
\qquad{}\qquad{}Std dev & - & - & \{0.0022\} & \{0.0052\} & \{0.0077\} & \{0.0057\}\tabularnewline
\hline 
\qquad{}\qquad{}Standard error & - & - & - & - & - & (0.0022)\tabularnewline
\hline 
Third moment ($\mu_{3}$) &  &  &  &  &  & \tabularnewline
\hline 
\qquad{}Plug-in $\left(\hat{\mu}_{3,PI}\right)$ & -0.0090 & -0.0055 & -0.0051 & -0.0049 & -0.0039 & -0.0038\tabularnewline
\hline 
\qquad{}Branch estimate $\left(\hat{\mu}_{3}\right)$ &  &  &  &  &  & \tabularnewline
\hline 
\qquad{}\qquad{}Mean & - & - & - & -0.0085 & -0.0115 & -0.0016\tabularnewline
\hline 
\qquad{}\qquad{}Std dev & - & - & - & \{0.0029\} & \{0.0037\} & \{0.0039\}\tabularnewline
\hline 
Fourth moment ($\mu_{4}$) &  &  &  &  &  & \tabularnewline
\hline 
\qquad{}Plug-in $\left(\hat{\mu}_{4,PI}\right)$ & 0.0176 & 0.0084 & 0.0070 & 0.0064 & 0.0054 & 0.0051\tabularnewline
\hline 
\qquad{}Branch estimate $\left(\hat{\mu}_{4}\right)$ &  &  &  &  &  & \tabularnewline
\hline 
\qquad{}\qquad{}Mean & - & - & - & - & - & 0.0036\tabularnewline
\hline 
\qquad{}\qquad{}Std dev & - & - & - & - & - & \{0.0025\}\tabularnewline
\hline 
\end{tabular}
\par\end{centering}
\caption{Centered moments of firm effects (size-weighted)\protect\label{tab:Centered=000020Moments}}

{\footnotesize\bigskip{}
Notes: the rows labeled ``Mean'' report an average over 100 randomly
drawn packings. The number in curly braces is the standard deviation
of the branch-based estimates across these packings. Dividing this
number by 10 gives the standard error associated with repacking variability.
Number in parentheses reports a standard error capturing sampling
variability of the average estimated second moment (see \ref{sec:The-variance-of}
for details).}{\footnotesize\par}
\end{table}

The plug-in estimates of the second moment of firm effects are upward
biased due to sampling error \citep{andrews2008high}. A useful benchmark
comes from \citet{kline2020leave}, who report a plug-in estimate
of 0.0358 and a bias-corrected estimate of 0.0240 in a closely related
sample, implying an upward bias of nearly 50\%.\footnote{The \citet{kline2020leave} estimates pertain to a sample under which
the mobility graph remains connected whenever any individual \textit{mover}
is removed. This sample can be thought of as lying between the $1$-ECC
and the $2$-ECC as each undirected edge may involve one or more movers.
One should expect the bias to be larger in less connected samples
\citep{jochmans2019fixed}.} Table \ref{tab:Centered=000020Moments} finds somewhat smaller biases
in the $k$-ECCs for which branch-based estimates can be computed.
For example, in the $3$-ECC, we find $\hat{\mu}_{2,PI}-\hat{\mu}_{2}=0.0063$,
implying a 22\% upward bias in the plug-in estimate. 

On average, the biases tend to grow smaller as the network grows more
connected. Note that for $k=5$ we find $\hat{\mu}_{2}\approx\hat{\mu}_{2,PI}$,
which may reflect noise in both the plug-in and covariance based estimates.
In general, the unbiased branch-based estimators need not respect
logical constraints on the parameter space. For instance, $\hat{\mu}_{2}$
can take on negative values. To rule such behavior out, one can consider
biased estimators of the form $\min\left\{ \max\left\{ 0,\hat{\mu}_{2}\right\} ,\hat{\mu}_{2,PI}\right\} $.

As described in \ref{sec:The-variance-of}, when four or more branches
are present, it becomes possible to estimate the variance of the second
moment estimator $\hat{\mu}_{2}$ by exploiting the covariance between
disjoint pairs of differences in the branch-specific estimates. Applying
this variance estimator to the $6$-ECC, and averaging across the
100 packings, yields an estimated $\mathbb{V}\left[\hat{\mu}_{2}\right]$
of $(0.0022)^{2}$. Thus, the sampling uncertainty in $\hat{\mu}_{2}$
turns out to be of the same order of magnitude as its packing uncertainty.
After averaging across 100 packings, however, the packing uncertainty
becomes negligible relative to sampling uncertainty. 

Fortunately, the ratio of the average estimate of $\hat{\mu}_{2}$
to its standard error yields a $z$-score of $0.0273/0.0022\approx12.4$,
suggesting that $\mu_{2}$ is estimated precisely. A similar finding
was reported by \citet[Table IV]{kline2020leave}, who obtained a
standard error of 0.0006 on a closely related unbiased estimator of
$\mu_{2}$, yielding a $z$-score of roughly $40$. The precision
with which $\hat{\mu}_{2}$ is estimated suggests that it can safely
be used to scale the estimates of higher-order moments without inducing
weak identification problems.

Table \ref{tab:Scaled=000020Moments} converts the average moment
estimates from Table \ref{tab:Centered=000020Moments} into scaled
moments, which facilitates comparison to a Gaussian benchmark. The
scaled higher moment estimates suggest that the size-weighted distribution
of firm effects exhibits a left skew: more workers are employed at
firms with very low wages than at firms with very high wages. This
negative skew is less pronounced in the $6$-ECC, perhaps because
firms in the $6$-ECC tend to be very large. 

\begin{table}[H]
\begin{centering}
\begin{tabular}{|l|c|c|c|c|c|c|}
\hline 
\qquad{}\qquad{}$k$ & 1 & 2 & 3 & 4 & 5 & 6\tabularnewline
\hline 
\hline 
Std dev ($\sigma=\sqrt{\mu_{2}}$) &  &  &  &  &  & \tabularnewline
\hline 
\qquad{}Plug-in & 0.2215 & 0.1951 & 0.1841 & 0.1794 & 0.1751 & 0.1756\tabularnewline
\hline 
\qquad{}Branch means & - & - & 0.1660 & 0.1562 & 0.1749 & 0.1652\tabularnewline
\hline 
Skew $\left(\mu_{3}/\sigma^{3}\right)$ &  &  &  &  &  & \tabularnewline
\hline 
\qquad{}Plug-in & -0.8332 & -0.7403 & -0.8223 & -0.8507 & -0.7198 & -0.6998\tabularnewline
\hline 
\qquad{}Branch means & - & - & - & -2.2307 & -2.1482 & -0.3567\tabularnewline
\hline 
Kurtosis $\left(\mu_{4}/\sigma^{4}\right)$ &  &  &  &  &  & \tabularnewline
\hline 
\qquad{}Plug-in & 7.3087 & 5.8237 & 6.1367 & 6.1908 & 5.7763 & 5.4156\tabularnewline
\hline 
\qquad{}Branch means & - & - & - & - & - & 4.8572\tabularnewline
\hline 
\end{tabular}
\par\end{centering}
\caption{Scaled moments of firm effects (size-weighted)\protect\label{tab:Scaled=000020Moments}}

{\footnotesize\bigskip{}
Notes: the rows labeled ``Branch means'' report transformations
of the entries in Table \ref{tab:Centered=000020Moments} corresponding
to averages across packings of the relevant unbiased estimators of
centered moments.}{\footnotesize\par}
\end{table}

The four branches of the $6$-ECC allow us to estimate the size-weighted
fourth moment. While a normal distribution would exhibit a kurtosis
of 3, our preferred branch-based estimate of kurtosis is 4.9, indicating
heavy tails. Evidently, there are more very low and very high paying
firms in the $6$-ECC than one would surmise based upon a Gaussian
benchmark.

\subsection{Shrinkage\protect\label{subsec:Shrinkage}}

It is often of interest, either for forecasting or policy targeting
purposes, to construct a best predictor of the firm effects $\psi$
based upon the full-sample estimates $\hat{\psi}$. Standard empirical
Bayes shrinkage arguments are predicated on the assumption that the
distribution of noise is known ex-ante \citep{waltersEB2024}. Given
that the number of movers per edge is only 1.3, we cannot appeal to
a central limit theorem to defend a Gaussian approximation to the
noise in the branch-specific estimates $\left\{ \hat{\psi}_{b}\right\} _{b=1}^{M}$.
It turns out, however, that independence across branches facilitates
the construction of optimal predictors via standard regression methods
that remain valid regardless of the noise distribution.

To understand the basic logic of this argument, consider the case
where $M=2$. Let $\hat{\psi}_{bj}$ denote the $j$th entry of $\hat{\psi}_{b}$
and $\psi_{j}$ the $j$th entry of $\psi$. Now consider a random
effects model wherein $\psi_{j}\overset{i.i.d}{\sim}G$ and $\hat{\psi}_{b}\mid\psi\sim F_{b}$,
with $F_{b}$ obeying $\mathbb{E}_{F_{b}}\left[\hat{\psi}_{b}\right]=\psi$.
By iterated expectations,
\[
\mathbb{E}\left[\hat{\psi}_{2j}\mid\hat{\psi}_{1j}\right]=\mathbb{E}\left[\mathbb{E}\left[\hat{\psi}_{2j}\mid\psi_{j},\hat{\psi}_{1j}\right]\mid\hat{\psi}_{1j}\right]=\mathbb{E}\left[\mathbb{E}\left[\hat{\psi}_{2j}\mid\psi_{j}\right]\mid\hat{\psi}_{1j}\right]=\mathbb{E}\left[\psi_{j}\mid\hat{\psi}_{1j}\right],
\]
where the second equality follows by independence of the branch measurement
errors. Thus, a regression of the estimates in one branch on those
from another yields a minimum mean squared error optimal predictor
of the latent firm effect regardless of the noise distributions $\left\{ F_{1},F_{2}\right\} $,
an insight that can be traced back at least to \citet{krutchkoff1967supplementary}.
Importantly, this result holds under arbitrary patterns of heteroscedasticity,
both across and within branches.

\begin{figure}[H]
\noindent\begin{minipage}[t]{1\columnwidth}%
\subfloat[{$\mathbb{E}\left[\hat{\psi}_{2j}\mid\hat{\psi}_{1j}\right]$}]{\includegraphics[scale=0.4]{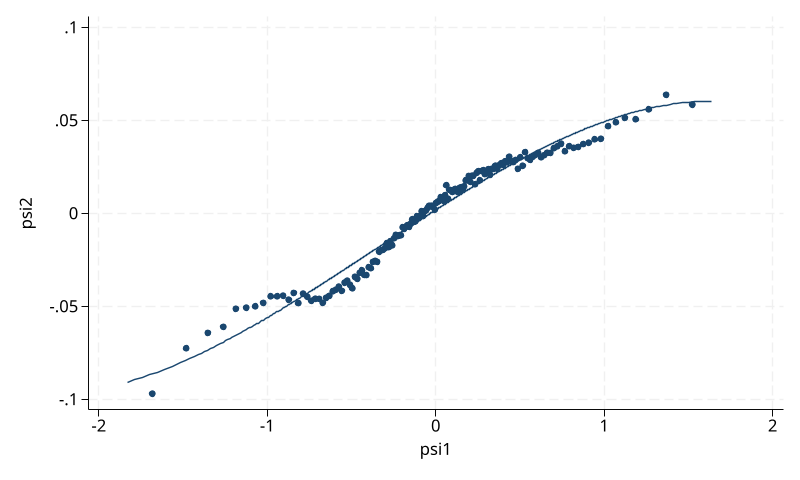}}\hfill{}\subfloat[{$\mathbb{E}\left[\hat{\psi}_{1j}\mid\hat{\psi}_{2j}\right]$}]{\includegraphics[scale=0.4]{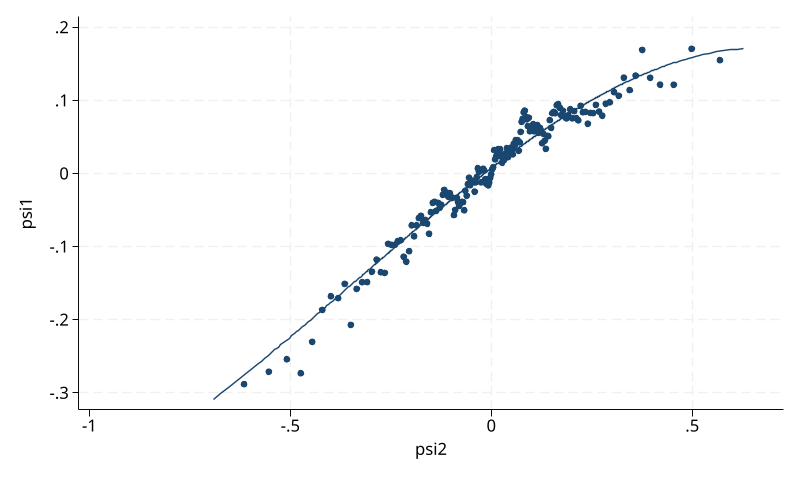}}%
\end{minipage}

\caption{Binscatters of polynomial fit to opposite branch\protect\label{fig:Binscatters}}
\end{figure}

Figure \ref{fig:Binscatters} illustrates this idea using the two
branches in the 3-ECC. Each panel depicts estimates of $\mathbb{E}\left[\hat{\psi}_{bj}\mid\hat{\psi}_{\ell j}=x\right]\equiv m_{b}\left(x\right)$
via the binscatter methods described in \citet{cattaneo2024binscatter},
where I have implicitly averaged over the 100 packings by applying
this method to a ``stacked'' dataset that concatenates $\left(\hat{\psi}_{1},\hat{\psi}_{2}\right)$
entries across packings. Note that the scale of the x-axis differs
across the two panels, reflecting that the entries of $\hat{\psi}_{2}$
are generally less variable than the corresponding entries of $\hat{\psi}_{1}$.
This is an artifact of the second branch utilizing leftover edges
that add additional information about the firm effects. 

In both panels of the figure, a third-order global polynomial approximation
$\hat{m}_{b}\left(x\right)$ to $m_{b}\left(x\right)$ is super-imposed
on the bins. This polynomial approximation fits well in the second
panel, where the predictor is the more precise $\hat{\psi}_{2j}$.
The fit is worse in the top panel, however, indicating that the conditional
expectation function is more complex when predicting with $\hat{\psi}_{1j}$. 

It is clear from the figure that the slope $\frac{d}{dx}\hat{m}_{b}\left(x\right)$
is less than one for both branches $b\in\left\{ 1,2\right\} $, which
is an example of the usual shrinkage phenomenon in optimal prediction
with noisy measurements. Rather than rely on one predictor or the
other, it makes sense to average them. For any $j$, an improved predictor
of $\psi_{j}$ is the average $\bar{m}\left(\hat{\psi}_{1j},\hat{\psi}_{2j}\right)=\frac{1}{2}\hat{m}_{1}\left(\hat{\psi}_{2j}\right)+\frac{1}{2}\hat{m}_{2}\left(\hat{\psi}_{1j}\right)$.

This idea generalizes to the case with more than two branches by regressing
the elements of each branch estimate on the corresponding entries
of all remaining branches. For example, with $M=3$, one would estimate
the functions $\mathbb{E}\left[\hat{\psi}_{b_{1}j}\mid\hat{\psi}_{b_{2}j}=x_{2},\hat{\psi}_{b_{3}j}=x_{3}\right]\equiv m_{b_{1}}\left(x_{2},x_{3}\right)$
across all three distinct triples of branches $\left(b_{1},b_{2},b_{3}\right)\in\left\{ 1,2,3\right\} :b_{1}\neq b_{2}\neq b_{3}$.
The resulting average prediction across branches can be written
\[
\bar{m}\left(\hat{\psi}_{1j},\hat{\psi}_{2j},\hat{\psi}_{3j}\right)=\frac{1}{3}\hat{m}_{1}\left(\hat{\psi}_{2j},\hat{\psi}_{3j}\right)+\frac{1}{3}\hat{m}_{2}\left(\hat{\psi}_{1j},\hat{\psi}_{3j}\right)+\frac{1}{3}\hat{m}_{3}\left(\hat{\psi}_{1j},\hat{\psi}_{2j}\right).
\]
Note that if the noise in the branches were identically distributed
it would necessarily be the case that $m_{1}\left(\cdot,\cdot\right)=m_{2}\left(\cdot,\cdot\right)=m_{3}\left(\cdot,\cdot\right)$.
In such a case, one might be tempted to collapse the left out branch
estimates down to their mean, and estimate the pooled leave-out conditional
expectation function $\mathbb{E}\left[\hat{\psi}_{b_{1}j}\mid\hat{\psi}_{b_{2}j}+\hat{\psi}_{b_{3}j}=x\right]$.

A recent paper by \citet{ignatiadis2023empirical} establishes that
it is possible to improve on such collapsed estimators when the noise
in each branch is identically distributed and non-Gaussian. Their
proposal involves regressing the estimates in each branch $b$ on
\textit{sorted values} of the estimates in all other branches. In
the case of $M=3$ branches, one would first regress $\hat{\psi}_{1j}$
on $\min\left\{ \hat{\psi}_{2j},\hat{\psi}_{3j}\right\} $ and $\max\left\{ \hat{\psi}_{2j},\hat{\psi}_{3j}\right\} $
to estimate the function 
\[
\mathbb{E}\left[\psi_{j}\mid\min\left\{ \hat{\psi}_{2j},\hat{\psi}_{3j}\right\} =\underline{x},\max\left\{ \hat{\psi}_{2j},\hat{\psi}_{3j}\right\} =\bar{x}\right]\equiv h_{1}\left(\underline{x},\bar{x}\right).
\]
Next, $\hat{\psi}_{2j}$ is regressed on $\min\left\{ \hat{\psi}_{1j},\hat{\psi}_{3j}\right\} $
and $\max\left\{ \hat{\psi}_{1j},\hat{\psi}_{3j}\right\} $ to estimate
$h_{2}\left(\underline{x},\bar{x}\right)$. Finally, $\hat{\psi}_{3j}$
is regressed on $\min\left\{ \hat{\psi}_{1j},\hat{\psi}_{2j}\right\} $
and $\max\left\{ \hat{\psi}_{1j},\hat{\psi}_{2j}\right\} $ to estimate
$h_{3}\left(\underline{x},\bar{x}\right)$. Averaging these three
cross-branch fits $\left(\hat{h}_{1},\hat{h}_{2},\hat{h}_{3}\right)$
yields the AURORA estimator $\bar{h}$, which \citet{ignatiadis2023empirical}
show performs nearly as well as an oracle that knows the noise distribution.
The rationale for using sorted values as regressors stems from the
observation that order statistics are sufficient for iid samples.
In the present setting, the noise distribution likely differs across
branches, which may undermine the advantages of relying on order statistics. 

\begin{table}[H]
\begin{centering}
\begin{tabular}{|c|c|c|c|}
\hline 
 & Naive $\left(\bar{\psi}\right)$ & Ignore order $\left(\bar{m}\right)$ & AURORA $\left(\bar{h}\right)$\tabularnewline
\hline 
\hline 
MSE & 0.173 & 0.044 & 0.044\tabularnewline
\hline 
\end{tabular}
\par\end{centering}
\caption{Predicting $\hat{\psi}_{4}$ using $\left(\hat{\psi}_{1},\hat{\psi}_{2},\hat{\psi}_{3}\right)$\protect\label{tab:Predicting--using}}

\bigskip{}

{\footnotesize Notes: all procedures applied to a ``stacked'' dataset
concatenating $\left(\hat{\psi}_{1},\hat{\psi}_{2},\hat{\psi}_{3},\hat{\psi}_{4}\right)$
entries across 100 alternate packings.}{\footnotesize\par}
\end{table}

Table \ref{tab:Predicting--using} summarizes the results of a forecasting
exercise in which the first three branch estimates $\left(\hat{\psi}_{1},\hat{\psi}_{2},\hat{\psi}_{3}\right)$
of the 6-ECC are used to form best predictors of $\psi$. I compute
both the AURORA estimator and the simpler regression based estimator
$\bar{m}\left(\hat{\psi}_{1j},\hat{\psi}_{2j},\hat{\psi}_{3j}\right)$
that ignores the order of the estimates. The conditional expectation
functions $\left\{ \hat{m}_{b}\right\} _{b=1}^{3}$ and $\left\{ \hat{h}_{b}\right\} _{b=1}^{3}$
underlying these approaches are computed via non-parametric series
regressions on B-spline bases tuned by cross-validation using Stata's
\texttt{npregress} function. To aid reproducibility, these steps are
again conducted in a stacked dataset that concatenates the data across
100 alternative packings. For comparison, I also report the naive
predictor $\bar{\psi}=\left(\hat{\psi}_{1}+\hat{\psi}_{2}+\hat{\psi}_{3}\right)/3$,
which serves as an unshrunk benchmark. To assess the quality of these
forecasts, the predictions are compared to $\hat{\psi}_{4}$, which
should be unbiased for $\psi$. For each predictor $\dot{\psi}\in\left\{ \bar{\psi},\bar{m},\bar{h}\right\} $,
the mean squared error of the prediction is computed as 
\[
MSE\left(\hat{\psi}_{4},\dot{\psi}\right)=\left(\hat{\psi}_{4}-\dot{\psi}\right)'\left(\hat{\psi}_{4}-\dot{\psi}\right)/J.
\]

As expected, the naive predictor incurs a large MSE because of the
noise in the branches. Relative to this benchmark, both the AURORA
estimator and its alternative that ignores order yield dramatic improvements
due to shrinkage. However, the order statistic approach does not appear
to convey any advantage over the simpler approach based on levels.
Whether the disappointing performance of AURORA is primarily attributable
to heteroscedasticity across branches or other factors (e.g., too
few branches or nearly Gaussian noise) is an interesting question
for future research.

\section{Conclusion}

As these examples illustrate, branches can dramatically simplify the
work of parsing signal from noise in fixed effects estimates. By quantifying
uncertainty, branches also enable advanced downstream tasks such as
moment estimation and shrinkage that form key aspects of the empirical
Bayes toolkit \citep{waltersEB2024}. 

More generally, breaking over-identified estimates down into simpler
estimates based upon branches is appealing on transparency grounds.
In this analysis, the firm effect estimates based on branches corresponding
to trees of the mobility graph are invertible linear transformations
of the oriented average wage changes $\overrightarrow{\Delta}_{b}$.
Therefore, all of the information in the microdata comprising a tree
is conveyed by the branch-specific estimate $\hat{\psi}_{b}$. In
addition to demystifying the origin of the full-sample estimates being
reported, branches may be useful for assessing the degree to which
over-identifying restrictions are violated in practice, a topic that
I leave for future work.

As our discussion of the Nash-Williams-Tutte theorem revealed, the
number of branches that can be extracted from a dataset depends crucially
on the edge connectivity of the mobility network. The findings reported
here suggest that a useful approximation to a $k$-ECC can be had
from a networks' $k$-core. In many settings (e.g., studies of migration,
bilateral trade, or friendship networks) the average degree of graph
vertices is very high, which should enable the extraction of dozens
of branches without substantial pruning of the network. The development
of algorithms capable of rapidly extracting all spanning trees from
enormous datasets with high degree is an interesting area for future
research.

\printbibliography

\pagebreak{}

\appendix
\setcounter{figure}{0}
\renewcommand{\thefigure}{A.\arabic{figure}}
\renewcommand{\thesection}{Appendix \Alph{section}\!}

\section{A random sampling interpretation of the AKM model\protect\label{sec:Independence-and-uncertainty}}

Consider the idealized case where there exists an infinite super-population
of movers between each of the origin-destination pairs in $\mathcal{P}$.
If we randomly sample movers from these populations then, for every
$\left(o,d\right)\in\mathcal{P}$,
\begin{equation}
y_{i2}-y_{i1}\mid\mathbf{j}\left(i,1\right)=o,\mathbf{j}\left(i,2\right)=d\overset{iid}{\sim}F_{od},\label{eq:iid}
\end{equation}
where $F_{od}:\mathbb{R}\rightarrow\left[0,1\right]$ is the distribution
of wage changes in the super-population of movers from firm $o$ to
firm $d$. The AKM model is a set of moment restrictions on the $\left\{ F_{od}\right\} _{\left(o,d\right)\in\mathcal{P}}$
stipulating that
\[
\int u\:dF_{od}\left(u\right)=\psi_{d}-\psi_{o},
\]
for all $\left(o,d\right)\in\mathcal{P}$.

In this thought experiment, the uncertainty that concerns us is how
estimates of the firm effects would change if a different set of movers
were sampled. The noise contributions can now be defined as departures
from the behavior of the average mover
\[
u_{i}:=y_{i2}-y_{i1}-\left(\psi_{\mathbf{j}\left(i,2\right)}-\psi_{\mathbf{j}\left(i,1\right)}\right).
\]
It follows that
\[
u_{i}\mid\mathbf{j}\left(i,1\right)=o,\mathbf{j}\left(i,2\right)=d\overset{iid}{\sim}H_{od},
\]
where $H_{od}\left(u\right)=F_{od}\left(u+\psi_{d}-\psi_{o}\right).$
Thus, when the AKM model holds, random sampling of movers necessarily
yields mutually independent (and identically distributed) noise contributions
$\left\{ u_{i}\right\} _{i=1}^{N}$, implying that Assumption \ref{Pair-independence}
is satisfied. 

In practice, AKM models are typically fit to administrative records
that involve no random sampling. However, one can view administrative
records as capturing a set of movers drawn randomly from a finite
population of potential movers. Suppose, for example, that there are
$n_{o}\in\mathbb{N}_{>0}$ workers at origin firm $o$ who share a
common probability $p_{od}>0$ of moving to firm $d$. Then the set
of movers observed to move over any two periods is a random draw from
the population $n_{o}$ of potential movers. When $n_{0}$ is large,
dependence between the random draws becomes asymptotically negligible
\citep{hajek1960limiting,serfling1980approximation}, implying \eqref{eq:iid}
will hold. Consequently, Assumption \ref{Pair-independence} is satisfied.

\pagebreak{}

\section{An example where greedy extraction fails\protect\label{sec:Appendix:-An-example}}

This appendix gives an example of how the order in which spanning
trees are extracted from a graph can influence the total number of
trees packed. The complete graph depicted in Figure \ref{fig:Two-attempts}
has 6 vertices and 15 edges. Since every vertex has degree 5, this
network is a $5$-ECC. By the Nash-Williams-Tutte theorem, the graph
must contain at least $\left\lfloor 5/2\right\rfloor =2$ trees. However,
the graph's actual tree packing number, $\tau\left(G\right)$, is
$3$.

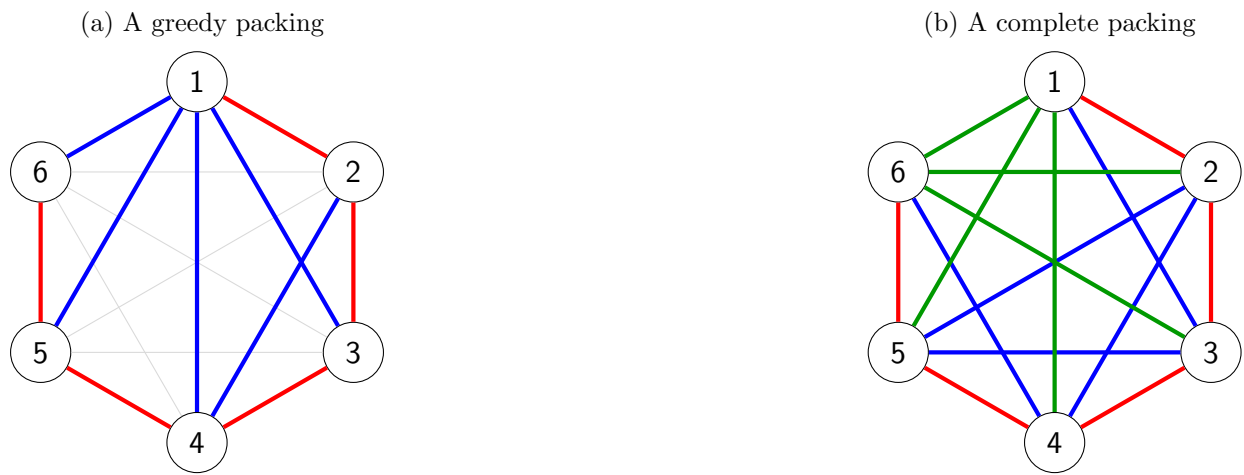
\begin{figure}[H]
\noindent\begin{minipage}[t]{1\columnwidth}%
\subfloat[A greedy packing]{\begin{tikzpicture}[scale=1.2,every node/.style={circle,draw,minimum size=8mm,font=\sffamily}]

% Panel a: greedy packing
\begin{scope}[shift={(0,0)}]
  % nodes
  \node (1a) at (90:2){1};
  \node (2a) at (30:2){2};
  \node (3a) at (-30:2){3};
  \node (4a) at (-90:2){4};
  \node (5a) at (-150:2){5};
  \node (6a) at (150:2){6};
  % draw all edges in gray
  \foreach \i/\j in {1a/2a,1a/3a,1a/4a,1a/5a,1a/6a,2a/3a,2a/4a,2a/5a,2a/6a,3a/4a,3a/5a,3a/6a,4a/5a,4a/6a,5a/6a}
    \draw[gray!30] (\i) -- (\j);
  % greedy tree T1 in red
  \foreach \i/\j in {1a/2a,2a/3a,3a/4a,4a/5a,5a/6a}
    \draw[red,ultra thick] (\i) -- (\j);
  % greedy tree T2 in blue
  \foreach \i/\j in {1a/3a,1a/4a,1a/5a,1a/6a,2a/4a}
    \draw[blue,ultra thick] (\i) -- (\j);
\end{scope}

\end{tikzpicture}

}\hfill{}\subfloat[A complete packing]{\begin{tikzpicture}[scale=1.2,every node/.style={circle,draw,minimum size=8mm,font=\sffamily}]

% Panel b: full packing
\begin{scope}[shift={(5,0)}]
  % nodes
  \node (1b) at (90:2){1};
  \node (2b) at (30:2){2};
  \node (3b) at (-30:2){3};
  \node (4b) at (-90:2){4};
  \node (5b) at (-150:2){5};
  \node (6b) at (150:2){6};
  % draw all edges in gray
  \foreach \i/\j in {1b/2b,1b/3b,1b/4b,1b/5b,1b/6b,2b/3b,2b/4b,2b/5b,2b/6b,3b/4b,3b/5b,3b/6b,4b/5b,4b/6b,5b/6b}
    \draw[gray!30] (\i) -- (\j);
  % spanning tree T1 in red
  \foreach \i/\j in {1b/2b,2b/3b,3b/4b,4b/5b,5b/6b}
    \draw[red,ultra thick] (\i) -- (\j);
  % spanning tree T2 in blue
  \foreach \i/\j in {1b/3b,3b/5b,5b/2b,2b/4b,4b/6b}
    \draw[blue,ultra thick] (\i) -- (\j);
  % spanning tree T3 in green
  \foreach \i/\j in {1b/4b,1b/5b,1b/6b,2b/6b,3b/6b}
    \draw[green!60!black,ultra thick] (\i) -- (\j);
\end{scope}

\end{tikzpicture}

}%
\end{minipage}

\caption{Two attempts to pack the same graph\protect\label{fig:Two-attempts}}
\end{figure}

The perils of greedy sequential extraction are depicted in panel (a),
which shows a possible choice of two initial spanning trees that have
been colored blue and red. The remaining edges fail to connect vertex
1, implying no remaining trees can be extracted after these two have
been removed. Panel (b) demonstrates that, in fact, three edge-disjoint
spanning trees that can be packed into the graph. While the three
shaded trees depicted are not unique, applying the Roskind-Tarjan
algorithm to this $5$-ECC with $M=3$ would ensure that three spanning
trees are found.\newpage{}

\section{Estimating the variance of $\hat{\mu}_{2}$\protect\label{sec:The-variance-of}}

The unbiased second central moment estimator described in Section
\ref{subsec:Moment-estimation} can be written
\begin{eqnarray*}
\hat{\mu}_{2} & = & \binom{M}{2}^{-1}\omega'\left(\sum_{b=1}^{M}\sum_{\ell<b}\hat{\mu}_{2,b,\ell}\right),
\end{eqnarray*}
where
\[
\hat{\mu}_{2,b,\ell}=\left(\hat{\psi}_{b}-\omega'\hat{\psi}_{b}\boldsymbol{1}\right)\odot\left(\hat{\psi}_{\ell}-\omega'\hat{\psi}_{\ell}\boldsymbol{1}\right).
\]

Define $\hat{g}_{b,\ell}=\omega'\hat{\mu}_{2,b,\ell}\in\mathbb{R}$,
so that
\[
\hat{\mu}_{2}=\binom{M}{2}^{-1}\left(\sum_{b=1}^{M}\sum_{\ell<b}\hat{g}_{b,\ell}\right).
\]
From independence across branches we have
\begin{eqnarray*}
\mathbb{V}\left[\hat{\mu_{2}}\right] & = & \binom{M}{2}^{-2}\mathbb{V}\left[\sum_{b=1}^{M}\sum_{\ell<b}\hat{g}_{b,\ell}\right]\\
 & = & \binom{M}{2}^{-2}\sum_{b=1}^{M}\sum_{\ell<b}\mathbb{V}\left[\hat{g}_{b,\ell}\right]\\
 & + & 2\binom{M}{2}^{-2}\sum_{b_{1}=1}^{M}\sum_{b_{2}<b_{1}}\sum_{b_{3}<b_{2}}\left(\mathbb{C}\left[\hat{g}_{b_{1},b_{2}},\hat{g}_{b_{1},b_{3}}\right]+\mathbb{C}\left[\hat{g}_{b_{1},b_{2}},\hat{g}_{b_{2},b_{3}}\right]+\mathbb{C}\left[\hat{g}_{b_{1},b_{3}},\hat{g}_{b_{2},b_{3}}\right]\right),
\end{eqnarray*}
where $\mathbb{C}$ denotes the covariance operator. This variance
expression separates into two parts: the variance of any single $\hat{g}_{b,\ell}$,
and the covariance of terms that share a single branch. 

The first sort of term can be estimated by examining the sample variance
of differences between disjoint pairs $\left(b_{1}>b_{2}\right)$
and $\left(b_{3}>b_{4}\right)$. Note that there are $6\binom{M}{4}$
ordered disjoint pairs, each of which obeys: 
\[
\mathbb{E}\left[\left(\hat{g}_{b_{1},b_{2}}-\hat{g}_{b_{3},b_{4}}\right)^{2}\right]=\mathbb{V}\left[\hat{g}_{b_{1},b_{2}}\right]+\mathbb{V}\left[\hat{g}_{b_{3},b_{4}}\right].
\]
In contrast, the covariance terms can be estimated by examining the
sample covariance across partially overlapping pairs. Specifically,
there are $\binom{M}{3}$ terms obeying 
\[
\mathbb{E}\left[\left(\hat{g}_{b_{1},b_{2}}-\hat{g}_{b_{1},b_{3}}\right)^{2}\right]=\mathbb{V}\left[\hat{g}_{b_{1},b_{2}}\right]+\mathbb{V}\left[\hat{g}_{b_{1},b_{3}}\right]-2\mathbb{C}\left[\hat{g}_{b_{1},b_{2}},\hat{g}_{b_{1},b_{3}}\right],
\]
where $b_{1}>b_{2}>b_{3}$.

Now define the sample variance across ordered disjoint pairs:
\[
\bar{V}=\frac{1}{6\binom{M}{4}}\underset{\left(b_{1},b_{2}\right)\cap\left(b_{3},b_{4}\right)=\emptyset}{\sum_{b_{1}>b_{2},\,b_{3}>b_{4}}}\left(\hat{g}_{b_{1},b_{2}}-\hat{g}_{b_{3},b_{4}}\right)^{2}.
\]
As the $\binom{M}{4}$ term reveals, constructing $\bar{V}$ requires
$M\geq4$. The corresponding sum of squared differences across partially
overlapping pairs is:
\begin{eqnarray*}
\bar{C} & = & \binom{M}{3}^{-1}\sum_{b_{1}>b_{2}>b_{3}}\left\{ \left(\hat{g}_{b_{1},b_{2}}-\hat{g}_{b_{1},b_{3}}\right)^{2}+\left(\hat{g}_{b_{1},b_{2}}-\hat{g}_{b_{2},b_{3}}\right)^{2}+\left(\hat{g}_{b_{1},b_{3}}-\hat{g}_{b_{2},b_{3}}\right)^{2}\right\} .
\end{eqnarray*}

Algebraic manipulations reveal
\begin{eqnarray*}
\mathbb{E}\left[\bar{V}\right] & = & \frac{2\binom{M-2}{2}}{6\binom{M}{4}}\sum_{b=1}^{M}\sum_{\ell<b}\mathbb{V}\left[\hat{g}_{b,\ell}\right]\\
 & = & \frac{4}{M\left(M-1\right)}\sum_{b=1}^{M}\sum_{\ell<b}\mathbb{V}\left[\hat{g}_{b,\ell}\right].
\end{eqnarray*}
\begin{eqnarray*}
\mathbb{E}\left[\bar{C}\right] & = & \frac{2\left(M-2\right)}{\binom{M}{3}}\sum_{b=1}^{M}\sum_{\ell<b}\mathbb{V}\left[\hat{g}_{b,\ell}\right]\\
 & - & \frac{2}{\binom{M}{3}}\sum_{b_{1}=1}^{M}\sum_{b_{2}<b_{1}}\sum_{b_{3}<b_{2}}\left(\mathbb{C}\left[\hat{g}_{b_{1},b_{2}},\hat{g}_{b_{1},b_{3}}\right]+\mathbb{C}\left[\hat{g}_{b_{1},b_{2}},\hat{g}_{b_{2},b_{3}}\right]+\mathbb{C}\left[\hat{g}_{b_{1},b_{3}},\hat{g}_{b_{2},b_{3}}\right]\right).
\end{eqnarray*}

It follows that an unbiased estimator for the variance of $\hat{\mu}_{2}$
is:
\[
\hat{\mathbb{V}}\left[\hat{\mu_{2}}\right]=\binom{M}{2}^{-2}\left[\binom{M}{2}\frac{\bar{V}}{2}+\binom{M}{3}\left(3\bar{V}-\bar{C}\right)\right].
\]
While $\hat{\mathbb{V}}\left[\hat{\mu_{2}}\right]$ is unbiased for
$\mathbb{V}\left[\hat{\mu_{2}}\right]$, it is not guaranteed to yield
a non-negative estimate. One can remedy this by employing the constrained
estimator $\max\left\{ 0,\hat{\mathbb{V}}\left[\hat{\mu_{2}}\right]\right\} $,
which possesses a small upward bias.
\end{document}